\documentclass[a4paper,UKenglish,cleveref, autoref, thm-restate]{lipics-v2021}

\usepackage{german}
\usepackage[utf8]{inputenc}
\usepackage{graphicx}
\usepackage{algorithm}
\usepackage[noend]{algorithmic}

\usepackage{tabularx}
\usepackage{amsthm,amsmath,amssymb,amsfonts,cite}
\usepackage{mathtools}
\usepackage{makecell}

\usepackage{tikz}
\usetikzlibrary{arrows.meta}
\tikzset{
	vertex/.style={circle,draw,fill=black, minimum size = 5pt, inner sep=0pt},
	edge/.style={color=black, thick},
	diredge/.style={->,>={Stealth[width=8pt,length=10pt]},color=black, thick},
	timelabel/.style={fill=white},
	specialedge/.style={dashed,blue!50!white},
}

\usepackage{todonotes}

\theoremstyle{definition}

\crefname{rrule}{Data Reduction Rule}{Data Reduction Rules} 

\newcommand{\problemdef}[3]{
	\begin{center}
	\begin{minipage}{0.95\textwidth}
		\noindent
		#1
		\vspace{5pt}\\
		\setlength{\tabcolsep}{3pt}
		\begin{tabularx}{\textwidth}{@{}lX@{}}
			\textbf{Input:}     & #2 \\
			\textbf{Question:}  & #3
		\end{tabularx}
	\end{minipage}
	\end{center}
}

\newcommand{\tn}{\textnormal}

\DeclarePairedDelimiterX{\abs}[1]{\lvert}{\rvert}{#1}
\newcommand{\mvert}{\;\middle\vert\;}

\newcommand{\TG}{\mathcal{G}}

\newcommand{\bigO}{{O}}

\newcommand{\no}{\emph{no}}
\newcommand{\yes}{\emph{yes}}

\newcommand{\true}{\texttt{true}}
\newcommand{\false}{\texttt{false}}

\newcommand{\maxtwosat}{\textsc{Max-2-Sat}}

\newcommand{\nae}{\texttt{NAE}}

\newcommand{\NAESat}{\textsc{Not-All-Equal-3Sat}}

\title{The Complexity of Transitively Orienting Temporal Graphs} 

\newcommand{\tuaddress}{Technische Universit"at Berlin, Faculty IV, Algorithmics and Computational Complexity, Germany}

\author{George B. Mertzios}
{Department of Computer Science, Durham University, UK}
{george.mertzios@durham.ac.uk}
{https://orcid.org/0000-0001-7182-585X}{Supported by the EPSRC grant EP/P020372/1.}

\author{Hendrik Molter}
{Department of Industrial Engineering and Management, Ben-Gurion University of the Negev, 
Israel}
{molterh@post.bgu.ac.il}
{https://orcid.org/0000-0002-4590-798X}{Supported by the German Research
Foundation (DFG), project MATE (NI 369/17), and by the Israeli Science Foundation (ISF), grant No.~1070/20.}

\author{Malte Renken}
{\tuaddress}
{m.renken@tu-berlin.de}
{http://orcid.org/0000-0002-1450-1901}{Supported by the German Research
Foundation (DFG), project MATE (NI 369/17).}

\author{Paul G. Spirakis}
{Department of Computer Science, University of Liverpool, UK\\
Computer Engineering \& Informatics Department, University of Patras, Greece}
{p.spirakis@liverpool.ac.uk}
{https://orcid.org/0000-0001-5396-3749}{Supported by the NeST initiative of the School of EEE and CS at the University of Liverpool and by the EPSRC grant EP/P02002X/1.}

\author{Philipp Zschoche}
{\tuaddress}
{zschoche@tu-berlin.de}
{https://orcid.org/0000-0001-9846-0600}{}

\authorrunning{G.B.~Mertzios, H.~Molter, M.~Renken, P.G.~Spirakis and P.~Zschoche} 

\Copyright{G.B.~Mertzios, H.~Molter, M.~Renken, P.G.~Spirakis and P.~Zschoche} 



\nolinenumbers 
\hideLIPIcs

\begin{document}

\maketitle

\begin{abstract}
In a \emph{temporal network} with discrete time-labels on its edges, entities and information can only ``flow'' along sequences of edges 
whose time-labels are non-decreasing (resp.~increasing), i.e.~along temporal (resp.~strict temporal) paths. 
Nevertheless, in the model for temporal networks of [Kempe, Kleinberg, Kumar, JCSS, 2002], the individual time-labeled edges remain undirected: 
an edge $e=\{u,v\}$ with time-label $t$ specifies that ``$u$ communicates with $v$ at time $t$''. 
This is a symmetric relation between $u$ and $v$, and it can be interpreted that the information can flow in either direction. 

In this paper we make a first attempt to understand how the direction of information flow on one edge 
can impact the direction of information flow on other edges. 
More specifically, naturally extending the classical notion of a transitive orientation in static graphs, 
we introduce the fundamental notion of a \emph{temporal transitive orientation} and we systematically investigate its algorithmic behavior in various situations. 
An orientation of a temporal graph is called \emph{temporally transitive} if, 
whenever $u$ has a directed edge towards $v$ with time-label $t_1$ and $v$ has a directed edge towards $w$ with time-label $t_2\geq t_1$, 
then $u$ also has a directed edge towards $w$ with some time-label $t_3\geq t_2$. 
If we just demand that this implication holds whenever $t_2 > t_1$, the orientation is called \emph{strictly} temporally transitive, 
as it is based on the fact that there is a strict directed temporal path from $u$ to $w$. 

Our main result is a conceptually simple, yet technically quite involved, 
polynomial-time algorithm for recognizing whether a given temporal graph $\TG$ is transitively orientable. 
In wide contrast we prove that, surprisingly, it is NP-hard to recognize whether $\TG$ is strictly transitively orientable. 
Additionally we introduce and investigate further related problems to temporal transitivity, 
notably among them the \emph{temporal transitive completion} problem, for which we prove both algorithmic and hardness results.
\end{abstract}

\section{Introduction}\label{intro-sec}

A \emph{temporal} (or \emph{dynamic}) network is, roughly speaking, a network whose underlying topology changes over time. 
This notion concerns a great variety of both modern and traditional networks; information and communication networks, social networks, 
and several physical systems are only few examples of networks which change over time~\cite{michailCACM,Nicosia-book-chapter-13,holme2019temporal}. 
Due to its vast applicability in many areas, the notion of temporal graphs has been studied from 
different perspectives under several different names such as \emph{time-varying}, \emph{evolving}, 
\emph{dynamic}, and \emph{graphs over time} (see~\cite{CasteigtsFloccini12,flocchini1,flocchini2} and the references therein).
In this paper we adopt a simple and natural model for temporal networks which is given with discrete time-labels on the edges of a graph, 
while the vertex set remains unchanged. This formalism originates in the foundational work of Kempe et al.~\cite{KKK00}.

\begin{definition}[Temporal Graph~\cite{KKK00}]
\label{temp-graph-def} A \emph{temporal graph} is a pair $\TG=(G,\lambda)$,
where $G=(V,E)$ is an underlying (static) graph and $\lambda :E\rightarrow \mathbb{N}$ is a \emph{time-labeling} function which assigns to every
edge of $G$ a discrete-time label.
\end{definition}

Mainly motivated by the fact that, due to causality, entities and information in temporal graphs 
can only ``flow'' along sequences of edges whose time-labels are non-decreasing (resp.~increasing), 
Kempe et al.~introduced the notion of a \emph{(strict) temporal path}, or \emph{(strict) time-respecting path}, 
in a temporal graph $(G,\lambda)$ as a path in $G$ with edges $e_1,e_2,\ldots,e_k$ such that 
$\lambda(e_1)\leq \ldots \leq \lambda(e_k)$ (resp.~$\lambda(e_1)< \ldots < \lambda(e_k)$). 
This notion of a temporal path naturally resembles the notion of a \emph{directed} path in the classical static graphs, 
where the direction is from smaller to larger time-labels along the path. 
Nevertheless, in temporal paths the individual time-labeled edges remain undirected: 
an edge $e=\{u,v\}$ with time-label $\lambda(e)=t$ can be abstractly interpreted as ``$u$ communicates with $v$ at time $t$''. 
Here the relation ``communicates'' is symmetric between $u$ and $v$, i.e.~it can be interpreted that the information can flow in either direction.

In this paper we make a first attempt to understand how the direction of information flow on one edge 
can impact the direction of information flow on other edges. 
More specifically, naturally extending the classical notion of a transitive orientation in static graphs~\cite{Golumbic04}, 
we introduce the fundamental notion of a \emph{temporal transitive orientation} and we thoroughly investigate its algorithmic behavior in various situations. 
Imagine that $v$ receives information from $u$ at time $t_1$, while $w$ receives information from $v$ at time $t_2\geq t_1$. 
Then $w$ \emph{indirectly} receives information from $u$ through the intermediate vertex $v$. 
Now, if the temporal graph correctly records the transitive closure of information passing, the directed edge from $u$ to $w$ must exist 
and must have a time label $t_3\geq t_2$. In such a \emph{transitively oriented} temporal graph, 
whenever an edge is oriented from a vertex $u$ to a vertex $w$ with time-label $t$, 
we have that \emph{every} temporal path from $u$ to $w$ arrives no later than~$t$, and that there is no temporal path from $w$ to $u$. 
Different notions of temporal transitivity have also been used for automated
temporal data mining~\cite{moskovitch2015fast} in medical
applications~\cite{moskovitch2009medical}, text processing~\cite{tannier2011evaluating}.
Furthermore, in behavioral ecology, researchers have used a notion of orderly (transitive) triads A-B-C to quantify dominance among species. 
In particular, animal groups usually form dominance hierarchies in which dominance relations are transitive and can also change 
with time~\cite{mcdonald2013comparative}.

One natural motivation for our temporal transitivity notion may come from applications where confirmation and verification of information is vital, 
where vertices may represent entities such as investigative journalists or police detectives who gather sensitive information. 
Suppose that $v$ queried some important information from $u$ (the information source) at time $t_1$, 
and afterwards, at time $t_2\geq t_1$, $w$ queried the important information from $v$ (the intermediary). 
Then, in order to ensure the validity of the information received, $w$ might want to verify it by \emph{subsequently} querying 
the information directly from $u$ at some time $t_3\geq t_2$. 
Note that $w$ might first receive the important information from $u$ through various other intermediaries, 
and using several channels of different lengths. Then, to maximize confidence about the information, 
$w$ should query $u$ for verification only after receiving the information from the latest of these indirect channels.

It is worth noting here that the model of temporal graphs given in~\cref{temp-graph-def} has been also used in its extended form,
in which the temporal graph may contain multiple time-labels per edge~\cite{MertziosMCS13}. 
This extended temporal graph model has been used to investigate temporal
paths~\cite{wu_efficient_2016,himmel_efficient_2019,xuan_computing_2003,CasteigtsHMZ20,MertziosMCS13,AkridaMNRSZ19}
and other temporal path-related notions such as temporal analogues of 
distance and diameter~\cite{AkridaGMS16}, 
reachability~\cite{AkridaGMS17} and
exploration~\cite{AkridaMNRSZ19,AkridaGMS16,Erlebach0K15,enright2021assigning},
separation~\cite{Flu+20,Zsc+20,KKK00}, and path-based centrality
measures~\cite{kim_temporal_2012,BMNR20}, as well as recently non-path problems
too such as temporal variations of coloring~\cite{MertziosMZcoloring19}, vertex
cover~\cite{AkridaMSZ18vertex-cover}, matching~\cite{MertziosMNZZmatching20},
cluster editing~\cite{Che+18}, and maximal
cliques~\cite{ViardLM16,Him+17,Ben+19}.
However, in order to better investigate and illustrate the inherent combinatorial structure of temporal transitivity orientations, 
in this paper we mostly follow the original definition of temporal graphs given by Kempe et al.~\cite{KKK00} 
with one time-label per edge~\cite{CasteigtsPS19,EnrightMMZ21,AxiotisF16}. 
Throughout the paper, whenever we assume multiple time-labels per edge we will state it explicitly; 
in all other cases we consider a single label per edge.

In static graphs, the transitive orientation problem has received extensive attention which resulted in numerous efficient algorithms. 
A graph is called \emph{transitively orientable} (or a \emph{comparability} graph) if it is possible to orient its edges such that, 
whenever we orient $u$ towards $v$ and $v$ towards $w$, then the edge between $u$ and $w$ exists and is oriented towards~$w$. 
The first polynomial-time algorithms for recognizing whether a given (static) graph $G$ on $n$ vertices and $m$ edges 
is comparability (i.e.~transitively orientable) were based on the notion of \emph{forcing} an orientation and had running time $O(n^3)$ 
(see Golumbic~\cite{Golumbic04} and the references therein). 
Faster algorithms for computing a transitive orientation of a given comparability graph have been later developed, 
having running times $O(n^2)$~\cite{Spinrad85} and $O(n+m\log n)$~\cite{McConnellS94}, 
while the currently fastest algorithms run in linear $O(n+m)$ time and are based on efficiently computing a modular decomposition 
of $G$~\cite{McConnellS99,McConnellS97}; see also Spinrad~\cite{Spinrad03}.
It is fascinating that, although all the latter algorithms compute a valid transitive orientation if $G$ is a comparability graph,
they fail to recognize whether the input graph is a comparability graph; instead they produce an orientation which is non-transitive if $G$ is not a 
comparability graph. The fastest known algorithm for determining whether a given orientation is transitive requires matrix multiplication, 
currently achieved in $O(n^{2.37286})$ time~\cite{alman2021refined}.

\subparagraph{Our contribution.}
In this paper we introduce the notion of \emph{temporal transitive orientation} 
and we thoroughly investigate its algorithmic behavior in various situations. 
An orientation of a temporal graph $\TG=(G,\lambda)$ is called \emph{temporally transitive} if, 
whenever $u$ has a directed edge towards $v$ with time-label $t_1$ 
and $v$ has a directed edge towards $w$ with time-label $t_2\geq t_1$,
then $u$ also has a directed edge towards $w$ with some time-label $t_3\geq t_2$. 
If we just demand that this implication holds whenever $t_2 > t_1$, the orientation is called \emph{strictly} temporally transitive, 
as it is based on the fact that there is a strict directed temporal path from $u$ to $w$. 
Similarly, if we demand that the transitive directed edge from $u$ to $w$ has time-label $t_3 >t_2$, the orientation is called 
\emph{strongly} (resp.~\emph{strongly strictly)} temporally transitive.

Although these four natural variations of a temporally transitive orientation seem superficially similar to each other, it turns out that 
their computational complexity (and their underlying combinatorial structure) varies massively. 
Indeed we obtain a surprising result in~\Cref{recognition-sec}: deciding whether a temporal graph $\TG$ admits a \emph{temporally transitive} 
orientation is solvable in polynomial time (\Cref{algorithm-TTO-subsec}), while it is NP-hard to decide whether it admits a \emph{strictly temporally transitive} orientation (\Cref{NP-hard-StrictTTO-subsec}). 
On the other hand, it turns out that, deciding whether $\TG$ admits a \emph{strongly} or a \emph{strongly strictly} temporal transitive orientation 
is (easily) solvable in polynomial time as they can both be reduced to 2SAT satisfiability.

Our main result is that, given a temporal graph $\TG=(G,\lambda)$, we can decide in polynomial time whether $\TG$ is 
transitively orientable, and at the same time we can output a temporal transitive orientation if it exists. 
Although the analysis and correctness proof of our algorithm is technically quite involved, 
our algorithm is simple and easy to implement, as it is based on the notion of \emph{forcing} an orientation.\footnote{That is, 
	orienting an edge from $u$ to $v$ \emph{forces} us to orient another edge from $a$ to $b$.} 
Our algorithm extends and generalizes the classical polynomial-time algorithm for computing 
a transitive orientation in static graphs described by Golumbic~\cite{Golumbic04}. 
The main technical difficulty in extending the algorithm from the static to the temporal setting is that, 
in temporal graphs we cannot simply use orientation forcings to eliminate the condition that a \emph{triangle} is not allowed to be cyclically oriented. 
To resolve this issue, we first express the recognition problem of temporally transitively orientable graphs as a Boolean satisfiability problem 
of a \emph{mixed} Boolean formula $\phi_{\text{3NAE}} \wedge \phi_{\text{2SAT}}$. 
Here $\phi _{\text{3NAE}}$ is a \textsc{3NAE} formula, i.e.,~the disjunction of clauses with three literals each, 
where every clause $\nae(\ell _{1}, \ell _{2}, \ell _{3})$ is satisfied if and only if at least one 
of the literals $\{\ell_{1},\ell _{2},\ell _{3}\}$ is equal to 1 and at least one of them is equal to 0. 
Note that every clause $\nae(\ell _{1}, \ell _{2}, \ell _{3})$ corresponds to the condition that a specific triangle in the temporal graph cannot 
be cyclically oriented. 
Furthermore $\phi _{\text{2SAT}}$ is a 2SAT formula, i.e.,~the disjunction of 2CNF clauses with two literals each, 
where every clause $(\ell _{1}\vee \ell _{2})$ is satisfied if and only if at least one of the 
literals $\{\ell _{1},\ell _{2}\}$ is equal to 1. 
However, although deciding whether $\phi _{\text{2SAT}}$ is satisfiable can be done
in linear time with respect to the size of the formula~\cite{aspvall1979linear}, the
problem Not-All-Equal-3-SAT is NP-complete~\cite{Schaefer78}.

In the second part of our paper (\Cref{completion-sec}) we consider a natural extension of the temporal orientability problem, 
namely the \emph{temporal transitive completion} problem. 
In this problem we are given a (partially oriented) temporal graph $\TG$ and a natural number $k$, 
and the question is whether it is possible to add at most $k$ new edges (with the corresponding time-labels) to $\TG$ such that the resulting 
temporal graph is (strongly/strictly/strongly strictly) transitively orientable. 
We prove that all four versions of temporal transitive completion are NP-complete, even when the input temporal graph is completely unoriented. 
In contrast we show that, if the input temporal graph $\TG$ is \emph{directed} (i.e.~if every time-labeled edge has a fixed orientation) then 
all versions of temporal transitive completion are solvable in polynomial time. 
As a corollary of our results it follows that all four versions of temporal transitive completion are fixed-parameter-tractable (FPT) 
with respect to the number $q$ of unoriented time-labeled edges in $\TG$.

In the third and last part of our paper (\Cref{multilayer-sec}) we consider the \emph{multilayer transitive orientation} problem. 
In this problem we are given an undirected temporal graph $\TG=(G,\lambda)$, where $G=(V,E)$, 
and we ask whether there exists an orientation $F$ of its edges (i.e.~with exactly one orientation for each edge of $G$) such that, 
for every `time-layer'' $t\geq 1$, the (static) oriented graph induced by the edges having time-label $t$ is transitively oriented in $F$. 
Problem definitions of this type are commonly referred to as multilayer problems~\cite{Bre+19}.
Observe that this problem trivially reduces to the static case if we assume that each edge has a single time-label,
as then each layer can be treated independently of all others.
However, if we allow $\TG$ to have multiple time-labels on every edge of~$G$, then we show that the problem becomes NP-complete, 
even when every edge has at most two labels.

\section{Preliminaries and Notation}\label{prelim-sec}

Given a (static) undirected graph $G=(V,E)$, an edge between two vertices $%
u,v\in V$ is denoted by the unordered pair $\{u,v\}\in E$, and in this case
the vertices $u,v$ are said to be \emph{adjacent}. If the graph is directed,
we will use the ordered pair $(u,v)$ (resp.~$(v,u)$) to denote the oriented
edge from $u$ to $v$ (resp.~from $v$ to $u$). For simplicity of the notation, we
will usually drop the parentheses and the comma when denoting an oriented
edge, i.e.~we will denote $(u,v)$ just by $uv$. Furthermore, $\widehat{uv}%
=\{uv,vu\}$ is used to denote the set of both oriented edges $uv$ and $vu$
between the vertices $u$ and $v$.

Let $S\subseteq E$ be a subset of the edges of an undirected (static) graph $G=(V,E)$, 
and let $\widehat{S} =\{uv,vu : \{u,v\}\in S\}$ be the set of both possible orientations $uv$ and $vu$ of every edge $\{u,v\}\in S$. 
Let $F\subseteq \widehat{S}$. If $F$ contains \emph{at least one} of the two possible orientations $uv$ and $vu$ of each edge $\{u,v\}\in S$, 
then $F$ is called an \emph{orientation} of the edges of $S$. 
$F$ is called a \emph{proper orientation} if it contains \emph{exactly one} 
of the orientations $uv$ and $vu$ of every edge $\{u,v\}\in S$. 
Note here that, in order to simplify some technical proofs, the above definition of an orientation 
allows $F$ to be not proper, i.e.~to contain \emph{both} $uv$ and $vu$ for a specific edge $\{u,v\}$. 
However, whenever $F$ is not proper, this means that $F$ can be discarded as it cannot be used as a part of a (temporal) transitive orientation. 
For every orientation $F$ denote by $F^{-1}=\{vu:uv\in F\}$ the \emph{reversal} of $F$. 
Note that $F\cap F^{-1}=\emptyset $ if and only if $F$ is proper.

In a temporal graph $\TG= (G,\lambda)$, where $G=(V,E)$, whenever $\lambda(\{v,w\})=t$ (or simply $\lambda(v,w)=t$), 
we refer to the tuple $(\{v,w\},t)$ as a \emph{time-edge} of $\TG$. 
A triangle of $(G,\lambda )$ on the vertices $u,v,w$ is a \emph{synchronous triangle} if $\lambda (u,v)=\lambda (v,w)=\lambda (w,u)$. 
Let $G=(V,E)$ and let $F$ be a proper orientation of the whole edge set $E$. 
Then $(\TG,F)$, or $(G,\lambda,F)$, is a \emph{proper orientation} of the temporal graph $\TG$; 
for simplicity we may also write that $F$ is a proper orientation of $\TG$. 
A \emph{partial proper orientation} $F$ of a temporal graph $\TG=(G,\lambda)$ is an orientation of a subset of $E$. 
To indicate that the edge $\{u,v\}$ of a time-edge $(\{u,v\},t)$ is oriented from $u$ to $v$ 
(that is, $uv\in F$ in a (partial) proper orientation $F$), we use the term $((u,v),t)$, or simply $(uv,t)$. 
For simplicity we may refer to a (partial) proper orientation just as a (partial) orientation, whenever the term ``proper'' is clear from the context.

A static graph $G=(V,E)$ is a \emph{comparability graph} if there exists a proper orientation $F$ of~$E$ which is \emph{transitive}, 
that is, if $F\cap F^{-1}=\emptyset $ and $F^{2}\subseteq F$, where $F^{2}=\{uw:uv,vw\in F$ for some vertex $v\}$~\cite{Golumbic04}. 
Analogously, in a temporal graph $\TG= (G,\lambda)$, where $G=(V,E)$, we define a proper orientation $F$ of $E$ to be \emph{temporally transitive}, 
if:

\begin{center}\fbox{
		\begin{minipage}{0.95\textwidth}
			\noindent
			whenever $(uv,t_1)$ and $(vw,t_2)$ are oriented time-edges in $(\TG,F)$ such that $t_2 \geq t_1$, 
			there exists an oriented time-edge $(uw,t_3)$ in $(\TG,F)$, for some $t_3\geq t_2$.
	\end{minipage}}
\end{center}

In the above definition of a temporally transitive orientation, if we replace the condition ``$t_3\geq t_2$'' with ``$t_3> t_2$'', 
then $F$ is called \emph{strongly temporally transitive}. 
If we instead replace the condition ``$t_2\geq t_1$'' with ``$t_2> t_1$'',
then $F$ is called \emph{strictly temporally transitive}.
If we do both of these replacements, then $F$ is called \emph{strongly strictly temporally transitive}.
Note that strong (strict) temporal transitivity implies (strict) temporal transitivity, while (strong) temporal transitivity implies (strong) strict temporal transitivity.
Furthermore, similarly to the established terminology for static graphs, we define a temporal graph $\TG= (G,\lambda)$, where $G=(V,E)$, 
to be a \emph{(strongly/strictly) temporal comparability graph} if there exists a proper orientation $F$ of $E$ 
which is \emph{(strongly/strictly) temporally transitive}.

\newcommand{\TTOs}{\textsc{TTO}}
\newcommand{\StrictTTOs}{\textsc{Strict TTO}}
\newcommand{\StrongTTOs}{\textsc{Strong TTO}}
\newcommand{\StrongStrictTTOs}{\textsc{Strong Strict TTO}}
\newcommand{\TTO}{\textsc{Temporal Transitive Orientation (\TTOs)}}
\newcommand{\StrictTTO}{\textsc{Strict Temporal Transitive Orientation (\StrictTTOs)}}
\newcommand{\StrongTTO}{\textsc{Strong Temporal Transitive (\StrongTTOs)}}
\newcommand{\StrongStrictTTO}{\textsc{Strong Strict Temporal Transitive Orientation (\StrongStrictTTOs)}}

We are now ready to formally introduce the following decision problem of recognizing whether 
a given temporal graph is temporally transitively orientable or not.

\problemdef{\TTO}
{A temporal graph $\TG=(G,\lambda)$, where $G=(V,E)$.}
{Does $\TG$ admit a temporally transitive orientation $F$ of $E$?\vspace{-0,2cm}}

In the above problem definition of \TTOs, if we ask for the existence of a strictly (resp.~strongly, or strongly strictly) 
temporally transitive orientation $F$, 
we obtain the decision problem \textsc{Strict} (resp.~\textsc{Strong}, or
\textsc{Strong Strict}) \TTO.

Let $\TG=(G,\lambda)$ be a temporal graph, where $G=(V,E)$. 
Let $G'=(V,E')$ be a graph such that $E\subseteq E'$, and let $\lambda' \colon E' \rightarrow \mathbb N$ be a time-labeling function such that 
$\lambda'(u,v)=\lambda(u,v)$ for every $\{u,v\}\in E$. Then the temporal graph $\TG'=(G',\lambda')$ is called a \emph{temporal supergraph of $\TG$}. 
We can now define our next problem definition regarding computing temporally orientable supergraphs of $\TG$.

\newcommand{\TTCs}{\textsc{TTC}}
\newcommand{\TTC}{\textsc{Temporal Transitive Completion (\TTCs)}}
\problemdef{\TTC}
{A temporal graph $\TG=(G,\lambda)$, where $G=(V,E)$, a (partial) orientation $F$ of $\TG$, and an integer~$k$.}
{Does there exist a temporal supergraph $\TG'=(G',\lambda')$ of $(G,\lambda)$, where $G'=(V,E')$, 
	and a transitive orientation $F' \supseteq F$ of $\TG'$ such that $|E' \setminus E|\le k$?}

Similarly to \TTOs, if we ask in the problem definition of \TTCs\ for the existence of a strictly (resp.~strongly, or strongly strictly) temporally transitive orientation $F'$, 
we obtain the decision problem \textsc{Strict} (resp.~\textsc{Strong}, or \textsc{Strong Strict}) \TTC.

Now we define our final problem which asks for an orientation $F$ of a temporal graph $\TG=(G,\lambda)$ 
(i.e.~with exactly one orientation for each edge of $G$) such that, 
for every ``time-layer'' $t\geq 1$, the (static) oriented graph defined by the edges having time-label $t$ is transitively oriented in $F$. 
This problem does not make much sense if every edge has exactly one time-label in $\TG$, as in this case it 
can be easily solved by just repeatedly applying any known static transitive orientation algorithm. 
Therefore, in the next problem definition, we assume that in the input temporal graph $\TG=(G,\lambda)$ every edge of $G$ 
potentially has multiple time-labels, i.e.~the time-labeling function is $\lambda :E\rightarrow 2^{\mathbb{N}}$.

\newcommand{\MTOs}{\textsc{MTO}}
\newcommand{\MTO}{\textsc{Multilayer Transitive Orientation (\MTOs)}}
\problemdef{\MTO}
{A temporal graph $\TG=(G,\lambda)$, where $G=(V,E)$ and $\lambda :E\rightarrow 2^{\mathbb{N}}$.}
{Is there an orientation $F$ of the edges of $G$ such that, for every $t\geq 1$, 
	the (static) oriented graph induced by the edges having time-label $t$ is transitively oriented?}

\section{The recognition of temporally transitively orientable graphs}\label{recognition-sec}

In this section we investigate the computational complexity of all variants of
\TTOs. We show that \TTOs\, as well as the two variants \StrongTTOs{} and 
\StrongStrictTTOs{}, are solvable in polynomial time, whereas \StrictTTOs{} turns
out to be NP-complete.

The main idea of our approach to solve \TTOs\ and its variants is to create
Boolean variables for each edge of the underlying graph $G$ and interpret setting a
variable to 1 or 0 as the two possible ways of directing the
corresponding edge.

More formally, for every edge $\{u,v\}$ we introduce a variable $x_{uv}$ and
setting this variable to 1 corresponds to the orientation $uv$ while
setting this variable to 0 corresponds to the orientation $vu$. 
Now consider the example of Figure~\ref{forcing-fig}(a), i.e.~an induced path of length two 
in the underlying graph $G$ on three vertices $u,v,w$, 
and let $\lambda(u,v)=1$ and $\lambda(v,w)=2$. 
Then the orientation $uv$ ``forces'' the orientation $wv$. 
Indeed, if we otherwise orient $\{v,w\}$ as $vw$, then the edge $\{u,w\}$ must exist and be oriented as $uw$ in any temporal transitive orientation, 
which is a contradiction as there is no edge between $u$ and $w$. 
We can express this ``forcing'' with the implication $x_{uv}\implies x_{wv}$.
In this way we can deduce the constraints that all triangles or induced paths
on three vertices impose on any (strong/strict/strong strict) temporal transitive orientation. 
We collect all these constraints in \Cref{table:2sat-triangle}.

\begin{table}[t]
	\centering
	\footnotesize
	\setcellgapes{5pt}
	\makegapedcells{}
	\begin{tabular}{r|ccc|cc}
		&\multicolumn{3}{c|}{\begin{tikzpicture}[yscale=1.5]
				\draw[every node/.style={vertex}]
				(0,0) node[label=left:$u$] (u) {}
				(2,0) node[label=right:$w$] (w) {}
				(1,1) node[label=above:$v$] (v) {}
				;
				\draw[edge,every node/.style={timelabel}]
				(u)
				-- node {$t_3$} (w)
				-- node {$t_2$} (v)
				-- node {$t_1$} (u)
				;
		\end{tikzpicture}}&\multicolumn{2}{c}{\begin{tikzpicture}[yscale=1.5]
				\draw[every node/.style={vertex}]
				(0,0) node[label=left:$u$] (u) {}
				(2,0) node[label=right:$w$] (w) {}
				(1,1) node[label=above:$v$] (v) {}
				;
				\draw[edge,every node/.style={timelabel}]
				(u)
				-- node {$t_1$} (v)
				-- node {$t_2$} (w)
				;
		\end{tikzpicture}}\\
		& $t_1=t_2=t_3$  & $ t_1 < t_2 = t_3$  & $t_1 \leq t_2 < t_3$
		& $t_1=t_2$  & $ t_1 < t_2$ \\
		\hline
		\TTOs{} & non-cyclic   & $wu = wv$     & \makecell{$vw \implies uw$\\ $vu \implies
			wu$} & \makecell{$uv = wv$} & \makecell{$uv \implies wv$} \\
		\StrongTTOs{}    & $\bot$       & $wu \land wv$ & \makecell{$vw \implies uw$\\ $vu \implies wu$} & \makecell{$uv = wv$} & \makecell{$uv \implies wv$} \\
		\StrictTTOs{}     & $\top$       & non-cyclic    & \makecell{$vw \implies uw$\\ $vu \implies wu$} & $\top$ & \makecell{$uv \implies wv$} \\
		\textsc{Str.\ Str.\ TTO}        & $\top$       & \makecell{$vu \implies wu$\\ $uv \implies wv$} & \makecell{$vw
			\implies uw$\\
			$vu \implies wu$} & $\top$ & \makecell{$uv \implies wv$} \\
	\end{tabular}
	\caption{Orientation conditions imposed by a triangle (left) and an induced path of
		length two (right) in the underlying graph $G$ for the decision problems
		(\textsc{Strict}/\textsc{Strong}/\textsc{Strong Strict}) \TTOs.
		Here, $\top$ means that no restriction is
		imposed, $\bot$ means that the graph is not orientable, and in the case of
		triangles, ``non-cyclic'' means that all orientations except the ones that
		orient the triangle cyclicly are allowed.
	}
	\label{table:2sat-triangle}
\end{table}

When looking at the conditions imposed on temporal transitive orientations 
collected in \Cref{table:2sat-triangle}, we can
observe that all conditions except ``non-cyclic'' are expressible in 2SAT.
Since 2SAT is solvable in linear time~\cite{aspvall1979linear}, it immediately follows that
the strong variants of temporal transitivity are solvable in polynomial time, as the next theorem states.
\begin{theorem}\label{thm:easytosee}
	\StrongTTOs{} and \StrongStrictTTOs{} are solvable in polynomial time.
\end{theorem}

In the variants \TTOs{} and \StrictTTOs{}, however, we can have triangles which 
impose a ``non-cyclic'' orientation of three edges (\Cref{table:2sat-triangle}). 
This can be naturally
modeled by a not-all-equal (NAE) clause.\footnote{A not all equal clause is a
	set of literals and it evaluates to \true\ if and only if at least two literals in the
	set evaluate to different truth values.} However, if we now na\"ively model the
conditions with a Boolean formula, we obtain a formula with 2SAT clauses and
3NAE clauses. Deciding whether such a formula is satisfiable is NP-complete in
general~\cite{Schaefer78}. Hence, we have to investigate these two variants more thoroughly.

The only difference between the triangles that impose these ``non-cyclic'' orientations in these two problem variants is that, 
in \TTOs{}, the triangle is \emph{synchronous} (i.e.~all its three edges have the same time-label), 
while in \StrictTTOs{} two of the edges are synchronous and the third one has a smaller time-label than the other two. 
As it turns out, this difference of the two problem variants has important implications on their computational complexity. 
In fact, we obtain a surprising result: \TTOs{} is solvable in polynomial time
while \StrictTTOs{} is NP-complete. 

In~\Cref{NP-hard-StrictTTO-subsec} we prove that \StrictTTOs{} is NP-complete 
and in~\Cref{algorithm-TTO-subsec} we provide our polynomial-time algorithm for \TTOs{}.

\subsection{\StrictTTOs\ is NP-Complete}\label{NP-hard-StrictTTO-subsec}

In this section we show that in contrast to the other variants, \StrictTTOs\ is
NP-complete.
\begin{theorem}\label{thm:sgehard}
	\StrictTTOs\ is NP-complete even if the temporal input graph has only four
	different time labels.
\end{theorem}
\begin{proof}
	We present a polynomial time reduction from
	\textsc{(3,4)-SAT}~\cite{tovey1984simplified} where, given a CNF formula~$\phi$
	where each clause contains exactly three literals and each variably appears in exactly four clauses, we are asked whether $\phi$ is
	satisfiable or not. Given a formula $\phi$, we construct a temporal graph $\TG$
	as follows.
	
	\smallskip
	
	\noindent \emph{Variable gadget.} For each variable $x$ that appears in $\phi$,
	we add eight vertices $a_x, a'_x, b_x, b'_x, c_x, c'_x, d_x, d'_x$ to $\TG$.
	We connect these vertices as depicted in \Cref{fig:reduction1}, that is, we add the
	following time edges to $\TG$: $(\{a_x, a'_x\},1)$, $(\{a'_x, b_x\},2)$, $(\{b_x,b'_x\},1)$, 
	$(\{b'_x, c_x\},2)$, $(\{c_x, c'_x\},1)$, $(\{c'_x, d_x\},2)$, $(\{d_x,d'_x\},1)$, $(\{d'_x, a_x\},2)$.
	
	\begin{figure}[t]
		\begin{center}
			\begin{tikzpicture}[scale=.75,yscale=.8]
				\draw[every node/.style={vertex,label position=below}]
				(0,0) node[label=$a_x$] (A) {}
				++(2,0) node[label=$a'_x$] (A2) {}
				++(2,0) node[label=$b_x$] (B)  {}
				++(2,0) node[label=$b'_x$] (B2) {}
				++(2,0) node[label=$c_x$] (C) {}
				++(2,0) node[label=$c'_x$] (C2) {}
				++(2,0) node[label=$d_x$] (D) {}
				++(2,0) node[label=$d'_x$] (D2) {}
				;
				
				\draw[edge,every node/.style={timelabel}]
				(A) -- node {1} (A2)
				-- node {2} (B)
				-- node {1} (B2)
				-- node {2} (C)
				-- node {1} (C2)
				-- node {2} (D)
				-- node {1} (D2)
				to[bend right] node {2} (A);
			\end{tikzpicture}
		\end{center}
		\caption{Illustration of the variable gadget used in the reduction in the
			proof of \Cref{thm:sgehard}.}
		\label{fig:reduction1}
	\end{figure}
	
	\smallskip
	
	\noindent \emph{Clause gadget.} For each clause $c$ of $\phi$,
	we add six vertices $u_c, u'_c, v_c, v'_c, w_c, w'_c$ to $\TG$.
	We connect these vertices as depicted in \Cref{fig:reduction2}, that is, we add the
	following time edges to $\TG$: $(\{u_c, u'_c\},2)$, $(v_c, v'_c\},1)$, $(\{w_c,w'_c\},2)$, 
	$(\{u_c, v_c\},2)$, $(\{v_c, w_c\},3)$, $(\{w_c, u_c\},3)$, $(\{v_c,w'_c\},3)$, $(\{w_c, v'_c\},3)$.
	
	\begin{figure}[t]
		\begin{center}
			\begin{tikzpicture}[scale=.6,yscale=1.1]
				\draw[every node/.style={vertex}]
				(-1.5,0)  node[label=below:$u_c$] (A) {}
				(1.5, 0)  node[label=below:$v_c$] (B) {}
				(0, 2)    node[label=left:$w_c$] (C) {}
				(-3,-1.5) node[label=below:$u'_c$] (A2) {}
				(3,-1.5)  node[label=below:$v'_c$] (B2) {}
				(0,4)     node[label=above:$w'_c$] (C2) {}
				;
				
				\node at (0,-2) {(a)};
				
				\draw[diredge, every node/.style={timelabel}]
				(A) edge node {2} (A2)
				(A) edge node	{2} (B)
				(C) edge node {3} (A)
				(B) edge node {3} (C)
				(B2) edge node {1} (B)
				(C2) edge node {2} (C)
				(C2) to[bend right] node {3} (A);
				\draw[diredge, every node/.style={timelabel}]
				(B2) to[bend right] node {3} (C);

				\draw[every node/.style={vertex}]
				(6.5,0)  node[label=below:$u_c$] (A21) {}
				(9.5, 0)  node[label=below:$v_c$] (B21) {}
				(8, 2)    node[label=left:$w_c$] (C21) {}
				(5,-1.5) node[label=below:$u'_c$] (A22) {}
				(11,-1.5)  node[label=below:$v'_c$] (B22) {}
				(8,4)     node[label=above:$w'_c$] (C22) {}
				;
				
				\node at (8,-2) {(b)};
				
				\draw[diredge, every node/.style={timelabel}]
				(A22) edge node {2} (A21)
				(B21) edge node	{2} (A21)
				(C21) edge node {3} (A21)
				(B21) edge node {3} (C21)
				(B21) edge node {1} (B22)
				(C22) edge node {2} (C21)
				(C22) to[bend right] node {3} (A21);
				\draw[diredge, every node/.style={timelabel}]
				(B22) to[bend right] node {3} (C21);
				
				\draw[every node/.style={vertex}]
				(14.5,0)  node[label=below:$u_c$] (A3) {}
				(17.5, 0)  node[label=below:$v_c$] (B3) {}
				(16, 2)    node[label=left:$w_c$] (C3) {}
				(13,-1.5) node[label=below:$u'_c$] (A23) {}
				(19,-1.5)  node[label=below:$v'_c$] (B23) {}
				(16,4)     node[label=above:$w'_c$] (C23) {}
				;
				
				\node at (16,-2) {(c)};
				
				\draw[diredge, every node/.style={timelabel}]
				(A23) edge node {2} (A3)
				(A3) edge node	{2} (B3)
				(C3) edge node {3} (A3)
				(C3) edge node {3} (B3)
				(B23) edge node {1} (B3)
				(C3) edge node {2} (C23)
				(C23) to[bend right] node {3} (A3);
				\draw[diredge, every node/.style={timelabel}]
				(B23) to[bend right] node {3} (C3);
			\end{tikzpicture}
		\end{center}
		\caption{Illustration of the clause gadget used in the reduction in the
			proof of \Cref{thm:sgehard} and three ways how to orient the
			edges in it.}
		\label{fig:reduction2}
	\end{figure}
	
	\smallskip
	
	\noindent \emph{Connecting variable gadgets and clause gadgets.} Let variable
	$x$ appear for the $i$th time in clause $c$ and let $x$ appear in the $j$th
	literal of $c$.
	The four vertex pairs $(a_x, a'_x), (b_x, b'_x), (c_x, c'_x), (d_x, d'_x)$ from the variable gadget of $x$
	correspond to the first, second, third, and fourth appearance of $x$,
	respectively. 
	The three vertices $u'_c, v'_c, w'_c$ correspond to the first, second, and third
	literal of $c$, respectively.
	Let $i=1$ and $j=1$. If $x$ appears non-negated, then we add the time edge
	$(\{a_x, u'_c\},4)$. Otherwise, if $x$ appears negated, we add the time edge
	$(\{a'_x, u'_c\},4)$. For all other values of $i$ and $j$ we add time edges
	analogously.
	
	This finishes the reduction. It can clearly be performed in polynomial time.
	
	\medskip
	
	\noindent \textbf{($\Rightarrow$):} Assume that we have a satisfying assignment
	for~$\phi$, then we can orient $\TG$ as follows. Then if a variable $x$ is set
	to \true, we orient the edges of the corresponding variable gadgets as follows:
	$(a_x, a'_x)$, $(b_x,a'_x)$, $(b_x, b'_x)$, $(c_x,b'_x)$, $(c_x,
	c'_x)$, $(d_x,c'_x)$, $(d_x,d'_x)$, $(a_x,d'_x)$. 
	Otherwise, if $x$ is set to false, we orient as follows: $(a'_x,a_x)$, $(a'_x,
	b_x)$, $(b'_x,b_x)$, $(b'_x, c_x)$, $(c'_x,c_x)$, $(c'_x, d_x)$, $(d'_x,d_x)$, $(d'_x, a_x)$. It
	is easy so see that both orientations are transitive.
	
	Now consider a clause in $\phi$ with literals $u,v,w$ corresponding to vertices
	$u'_c, v'_c, w'_c$ of the clause gadget, respectively. We have that at least one
	of the three literals satisfies the clause. If it is $u$, then we orient the
	edges in the clause gadgets as illustrated in \Cref{fig:reduction2}~(a). It
	is easy so see that this orientation is transitive.
	Furthermore, we orient the three edges connecting the clause gadgets to variable
	gadgets as follows: By construction the vertices $u'_c, v'_c, w'_c$ are each
	connected to a variable gadget. Assume, we have edges $\{u'_c,x\}, \{v'_c,y\},
	\{w'_c,z\}$. Then we orient as follows: $(x,u'_c), (v'_c,y),
	(w'_c,z)$, that is, we orient the edge connecting the literal that satisfies the
	clause towards the clause gadget and the other two edges towards the variable
	gadgets. This yields a transitive in the clause gadget.
	Note that the variable gadgets have time labels $1$ and $2$ so we can always orient the connecting
	edges (which have time label $4$) towards the variable gadget. We do this with
	all connecting edges except $(x,u'_c)$. This edge is oriented from the variable
	gadget towards the clause gadget, however it also corresponds to a literal that
	satisfies the clause. Then by construction, the edges incident to $x$ in the
	variable gadget are oriented away from $x$, hence our orientation is transitive. 
	
	Otherwise and if $v$ satisfies the clause,
	then we orient the edges in the clause gadgets as illustrated in \Cref{fig:reduction2}~(b).
	Otherwise (in this case $w$ has to satisfy the clause), we orient the
	edges in the clause gadgets as illustrated in \Cref{fig:reduction2}~(c). It
	is easy so see that each of these orientation is transitive. In both cases we
	orient the edges connecting the clause gadgets to the variable gadgets
	analogously to the first case discussed above. By analogous arguments we get
	that the resulting orientation is transitive.

	\medskip
	
	\noindent \textbf{($\Leftarrow$):} Note that all variable gadgets
	are cycles of length eight with edges having labels alternating between $1$ and
	$2$ and hence the edges have to also be oriented alternately. Consider the
	variable gadget corresponding to $x$. We interpret the orientation $(a_x, a'_x)$, $(b_x,a'_x)$, $(b_x, b'_x)$, $(c_x,b'_x)$, $(c_x,
	c'_x)$, $(d_x,c'_x)$, $(d_x,d'_x)$, $(a_x,d'_x)$ as setting $x$ to \true\ and we
	interpret the orientation $(a'_x,a_x)$, $(a'_x,
	b_x)$, $(b'_x,b_x)$, $(b'_x, c_x)$, $(c'_x,c_x)$, $(c'_x, d_x)$, $(d'_x,d_x)$, $(d'_x, a_x)$ as
	setting $x$ to true. We claim that this yields a satisfying assignment for
	$\phi$.
	
	Assume for contradiction that there is a clause $c$ in $\phi$ that is not
	satisfied by this assignment. Then by construction of the connection of variable
	gadgets and clause gadgets, the connecting edges have to be oriented towards the
	variable gadget in order to keep the variable gadget transitive. Let the three
	connecting edges be $\{u'_c,x\}, \{v'_c,y\},
	\{w'_c,z\}$ and their orientation $(u'_c,x), (v'_c,y),
	(w'_c,z)$. Then we have that $(u'_c,x)$ forces $(u'_c,u_c)$ which in turn forces
	$(w_c,u_c)$. We have that $(v'_c,y)$ forces $(v'_c,v_c)$ which in turn forces
	$(v_c,u_c)$. Furthermore, we now have that $(w_c,u_c)$ and $(v_c,u_c)$ force
	$(w_c,v_c)$. Lastly, we have that $(w'_c,z)$ forces $(w'_c,w_c)$ which in turn
	forces $(v_c,w_c)$, a contradiction to the fact that we forced $(w_c,v_c)$
	previously.
\end{proof}

\subsection{A polynomial-time algorithm for \TTOs}\label{algorithm-TTO-subsec}

Let $G=(V,E)$ be a static undirected graph. There are various
polynomial-time algorithms for deciding whether $G$ admits a transitive
orientation $F$. However our results in this section are inspired by the
transitive orientation algorithm described by Golumbic~\cite{Golumbic04},
which is based on the crucial notion of \emph{forcing} an orientation. The
notion of forcing in static graphs is illustrated in \Cref%
{forcing-fig}~(a): if we orient the edge $\{u,v\}$ as $uv$ (i.e.,~from $u$
to $v$) then we are forced to orient the edge $\{v,w\}$ as $wv$ (i.e.,~from $%
w $ to $v$) in any transitive orientation $F$ of $G$. Indeed, if we
otherwise orient $\{v,w\}$ as $vw$ (i.e.~from $v$ to $w$), then the edge $%
\{u,w\}$ must exist and it must be oriented as $uw$ in any transitive
orientation $F$ of $G$, which is a contradiction as $\{u,w\}$ is not an edge
of $G$. Similarly, if we orient the edge $\{u,v\}$ as $vu$ then we are
forced to orient the edge $\{v,w\}$ as $vw$. That is, in any transitive
orientation $F$ of $G$ we have that $uv\in F\Leftrightarrow wv\in F$. This
forcing operation can be captured by the binary forcing relation $\Gamma $
which is defined on the edges of a static graph $G$ as follows~\cite%
{Golumbic04}. 
\begin{equation}
	uv\ \Gamma \ u^{\prime }v^{\prime }\text{ \ \ \ if and only if \ \ \ }\left\{ 
	\begin{array}{l}
		\text{either }u=u^{\prime }\text{ and }\{v,v^{\prime }\}\notin E \\ 
		\text{or }v=v^{\prime }\text{ and }\{u,u^{\prime }\}\notin E%
	\end{array}%
	\right. .  \label{Gamma-static-def-eq}
\end{equation}%

\begin{figure}
	\centering
	\begin{tikzpicture}[yscale=1.5]
		\draw[every node/.style={vertex}]
		(0,0) node[label=left:$u$] (u) {}
		(2,0) node[label=right:$w$] (w) {}
		(1,1) node[label=above:$v$] (v) {}
		;
		\node at (1,-.4) {(a)};
		\draw[edge]
		(u)
		-- (v)
		-- (w)
		;
		
		\draw[every node/.style={vertex}]
		(6,0) node[label=left:$u$] (u2) {}
		(8,0) node[label=right:$w$] (w2) {}
		(7,1) node[label=above:$v$] (v2) {}
		;
		\node at (7,-.4) {(b)};
		\draw[edge,every node/.style={timelabel}]
		(u2)
		-- node {$3$} (w2)
		-- node {$5$} (v2)
		-- node {$5$} (u2)
		;
	\end{tikzpicture}
	\caption{The orientation $uv$ forces the orientation $wu$ and vice-versa in
		the examples of (a)~a static graph $G$ where $\{u,v\},\{v,w\}\in E(G)$ and $%
		\{u,w\}\protect\notin E(G)$, and of (b)~a temporal graph $(G,\protect\lambda %
		)$ where $\protect\lambda (u,w)=3<5=\protect\lambda (u,v)=\protect\lambda %
		(v,w)$.}
	\label{forcing-fig}
\end{figure}

We now extend the definition of $\Gamma $ in a natural way to the binary
relation $\Lambda $ on the edges of a temporal graph $(G,\lambda )$, see~%
\Cref{Lambda-def-eq}. For this, observe from %
\Cref{table:2sat-triangle} that the only cases, where we
have $uv\in F\Leftrightarrow wv\in F$ in any temporal transitive orientation
of $(G,\lambda )$, are when (i)~the vertices $u,v,w$ induce a path of length
2 (see \Cref{forcing-fig}~(a)) and $\lambda (u,v)=\lambda (v,w)$,
as well as when (ii)~$u,v,w$ induce a triangle and $\lambda (u,w)<\lambda
(u,v)=\lambda (v,w)$. The latter situation is illustrated in the example of
\Cref{forcing-fig}~(b). The binary forcing relation $\Lambda $ is
only defined on pairs of edges $\{u,v\}$ and $\{u^{\prime },v^{\prime }\}$
where $\lambda (u,v)=\lambda (u^{\prime },v^{\prime })$, as follows.%
\begin{equation}
	uv\ \Lambda \ u^{\prime }v^{\prime }\text{ \ if and only if \ }\lambda
	(u,v)=\lambda (u^{\prime },v^{\prime })=t\text{ and }\left\{ 
	\begin{array}{l}
		u=u^{\prime }\text{ and }\{v,v^{\prime }\}\notin E,  \text{or}\\
		v=v^{\prime }\text{ and }\{u,u^{\prime }\}\notin E, \text{or}\\
		u=u^{\prime }\text{ and }\lambda (v,v^{\prime })<t, \text{or}\\
		v=v^{\prime }\text{ and }\lambda (u,u^{\prime })<t.%
	\end{array}%
	\right.  \label{Lambda-def-eq}
\end{equation}%
Note that, for every edge $\{u,v\}\in E$ we have that $uv\ \Lambda \ uv$.$~$%
The forcing relation $\Lambda $ for temporal graphs shares some properties
with the forcing relation $\Gamma $ for static graphs. In particular, the
reflexive transitive closure $\Lambda ^{\ast }$ of $\Lambda $ is an
equivalence relation, which partitions the edges of each set $%
E_{t}=\{\{u,v\}\in E:\lambda (u,v)=t\}$ into its $\Lambda $\emph{%
	-implication classes} (or simply, into its \emph{implication classes}). Two
edges $\{a,b\}$ and $\{c,d\}$ are in the same $\Lambda $-implication class
if and only $ab\ \Lambda ^{\ast }\ cd$, i.e.~there exists a sequence%
\begin{equation*}
	ab=a_{0}b_{0}\ \Lambda \ a_{1}b_{1}\ \Lambda \ \ldots \ \Lambda \
	a_{k}b_{k}=cd\text{, with }k\geq 0\text{.}
\end{equation*}%
Note that, for this to happen, we must have $\lambda (a_{0},b_{0})=\lambda
(a_{1},b_{1})=\ldots =\lambda (a_{k},b_{k})=t$ for some $t\geq 1$. Such a
sequence is called a $\Lambda $-chain from $ab$ to $cd$, and we say that $ab$
(eventually) $\Lambda $-forces $cd$. Furthermore note that $ab\ \Lambda
^{\ast }\ cd$ if and only if $ba\ \Lambda ^{\ast }\ dc$.\ The next
observation helps the reader understand the relationship between the two
forcing relations $\Gamma $ and $\Lambda $.

\begin{observation}
	\label{Gamma-Lambda-obs}Let $\{u,v\}\in E$, where $\lambda (u,v)=t$, and let 
	$A$ be the $\Lambda $-implication class of $uv$ in the temporal graph $%
	(G,\lambda )$. Let $G^{\prime }$ be the static graph obtained by removing
	from $G$ all edges $\{p,q\}$, where $\lambda (p,q) < t$. Then $A$ is also
	the $\Gamma $-implication class of $uv$ in the static graph $G^{\prime }$.
\end{observation}

For the next lemma, we use the notation $\widehat{A}=\{uv,vu : uv \in A\}$.

\begin{lemma}
	\label{reverse-implication-classes-lem}Let $A$ be a $\Lambda $-implication
	class of a temporal graph $(G,\lambda )$.
	Then either $A=A^{-1}=\widehat{A}$ or $A\cap A^{-1}=\emptyset $.
\end{lemma}

\begin{proof}
	Suppose that $A\cap A^{-1}\neq \emptyset $, and let $uv\in A\cap A^{-1}$, i.e.~$uv,vu\in A$.
	Then, for any 
	$pq\in A$ we have that $pq\ \Lambda ^{\ast }\ uv$ and $qp\ \Lambda ^{\ast }\
	vu$. Since $\Lambda ^{\ast }$ is an equivalence relation and $uv,vu\in A$,
	it also follows that $pq,qp\in A$. Therefore also $pq,qp\in A^{-1}$, and
	thus $A=A^{-1}=\widehat{A}$.
\end{proof}

\begin{definition}
	\label{F-respects-A-def}
	Let $F$ be a proper orientation and $A$ be a $\Lambda$-implication class of a temporal graph $(G,\lambda)$. 
	If $A\subseteq F$, we say that $F$ \emph{respects} $A$.
\end{definition}

\begin{lemma}
	\label{disjoint-implication-classes-lem}
	Let $F$ be a proper orientation and $A$ be a $\Lambda$-implication class of a temporal graph $(G,\lambda)$. 
	Then $F$ respects either $A$ or $A^{-1}$ (i.e.~either $A\subseteq F$ or $A^{-1}\subseteq F$), and in either case $A\cap A^{-1}=\emptyset$.
\end{lemma}

\begin{proof}
	We defined the binary forcing relation $\Lambda $ to capture the fact that,
	for any temporal transitive orientation $F$ of $(G,\lambda )$, if $ab\
	\Lambda \ cd$ and $ab\in F$, then also $cd\in F$. Applying this property
	repeatedly, it follows that either $A\subseteq F$ or $F\cap A=\emptyset $.
	If $A\subseteq F$ then $A^{-1}\subseteq F^{-1}$. On the other hand, if $%
	F\cap A=\emptyset $ then $A\subseteq F^{-1}$, and thus also $A^{-1}\subseteq
	F$. In either case, the fact that $F\cap F^{-1}=\emptyset $ by the
	definition of a temporal transitive orientation implies that also $A\cap
	A^{-1}=\emptyset $.
\end{proof}

Let now $ab=a_{0}b_{0}\ \Lambda \ a_{1}b_{1}\ \Lambda \ \ldots \ \Lambda \
a_{k}b_{k}=cd$ be a given $\Lambda $-chain. Note by \Cref{Lambda-def-eq}
that, for every $i=1,\ldots ,k$, we have that either $a_{i-1}=a_{i}$ or $%
b_{i-1}=b_{i}$. Therefore we can replace the $\Lambda $-implication $%
a_{i-1}b_{i-1}\ \Lambda \ a_{i}b_{i}$ by the implications $a_{i-1}b_{i-1}\
\Lambda \ a_{i}b_{i-1}\ \Lambda \ a_{i}b_{i}$, since either $%
a_{i}b_{i-1}=a_{i-1}b_{i-1}$ or $a_{i}b_{i-1}=a_{i}b_{i}$. Thus, as this
addition of this middle edge is always possible in a $\Lambda $-implication,
we can now define the notion of a canonical $\Lambda $-chain, which always
exists.

\begin{definition}
	\label{canonical-chain-def}Let $ab\ \Lambda ^{\ast }\ cd$. Then any $\Lambda 
	$-chain of the from $$ab=a_{0}b_{0}\ \Lambda \ a_{1}b_{0}\ \Lambda \
	a_{1}b_{1}\ \Lambda \ \ldots \ \Lambda \ a_{k}b_{k-1}\ \Lambda \
	a_{k}b_{k}=cd$$ is a \emph{canonical} $\Lambda $-chain.
\end{definition}

The next lemma extends an important known property of the forcing relation
$\Gamma $ for static graphs~\cite[Lemma 5.3]{Golumbic04} to the temporal case.

\begin{lemma}[Temporal Triangle Lemma]
	\label{triangle-lemma}Let $(G,\lambda )$ be a temporal graph with a
	synchronous triangle on the vertices $a,b,c$, where $\lambda (a,b)=\lambda
	(b,c)=\lambda (c,a)=t$. Let $A,B,C$ be three $\Lambda $-implication classes
	of $(G,\lambda )$, where $ab\in C$, $bc\in A$, and $ca\in B$, where $A\neq
	B^{-1}$ and $A\neq C^{-1}$.
	
	\begin{enumerate}
		\item If some $b^{\prime }c^{\prime }\in A$, then $ab^{\prime }\in C$ and $%
		c^{\prime }a\in B$.
		
		\item If some $b^{\prime }c^{\prime }\in A$ and $a^{\prime }b^{\prime }\in C$,
		then $c^{\prime }a^{\prime }\in B$.
		
		\item No edge of $A$ touches vertex $a$.
	\end{enumerate}
\end{lemma}

\begin{proof}
	\begin{enumerate}
		\item Let $b^{\prime }c^{\prime }\in A$, and let $bc=b_{0}c_{0}\ \Lambda \
		b_{1}c_{0}\ \Lambda \ \ldots \ \Lambda \ b_{k}c_{k-1}\ \Lambda \
		b_{k}c_{k}=b^{\prime }c^{\prime }$ be a canonical $\Lambda $-chain from $bc$
		to $b^{\prime }c^{\prime }$. Thus note that all edges $b_{i}c_{i-1}$ and $%
		b_{i}c_{i}$ of this $\Lambda $-chain have the same time-label $t$ in $%
		(G,\lambda )$. We will prove by induction that $ab_{i}\in C$ and $c_{i}a\in
		B $, for every $i=0,1,\ldots ,k$. The induction basis follows directly by
		the statement of the lemma, as $ab=ab_{0}\in C$ and $ca=c_{0}a\in B$.
		
		Assume now that $ab_{i}\in C$ and $c_{i}a\in B$. If $b_{i+1}=b_{i}$ then
		clearly $ab_{i+1}\in C$ by the induction hypothesis. Suppose now that $%
		b_{i+1}\neq b_{i}$. If $\{a,b_{i+1}\}\notin E$ then $ac_{i}\ \Lambda \
		b_{i+1}c_{i}$. Then, since $c_{i}a\in B$ and $b_{i+1}c_{i}\in A$, it follows
		that $A=B^{-1}$, which is a contradiction to the assumption of the lemma.
		Therefore $\{a,b_{i+1}\}\in E$. Furthermore, since $b_{i}c_{i}\ \Lambda \
		b_{i+1}c_{i}$, it follows that either $\{b_{i},b_{i+1}\}\notin E$ or $%
		\lambda (b_{i},b_{i+1}) < t$. In either case it follows that $ab_{i}\
		\Lambda \ ab_{i+1}$, and thus $ab_{i+1}\in C$.
		
		Similarly, if $c_{i+1}=c_{i}$ then $c_{i+1}a\in B$ by the induction
		hypothesis. Suppose now that $c_{i+1}\neq c_{i}$. If $\{a,c_{i+1}\}\notin E$
		then $ab_{i+1}\ \Lambda \ c_{i+1}b_{i+1}$. Then, since $ab_{i+1}\in C$ and $%
		b_{i+1}c_{i+1}\in A$, it follows that $A=C^{-1}$, which is a contradiction
		to the assumption of the lemma. Therefore $\{a,c_{i+1}\}\in E$. Furthermore,
		since $b_{i+1}c_{i}\ \Lambda \ b_{i+1}c_{i+1}$, it follows that either $%
		\{c_{i},c_{i+1}\}\notin E$ or $\lambda (c_{i},c_{i+1}) < t$. In either
		case it follows that $c_{i}a\ \Lambda \ c_{i+1}a$, and thus $c_{i+1}a\in C$.
		This completes the induction step.
		
		\item Let $b^{\prime }c^{\prime }\in A$ and $a^{\prime }b^{\prime }\in C$.
		Then part 1 of the lemma implies that $c^{\prime }a\in B$. Now let $%
		ab=a_{0}b_{0}\ \Lambda \ a_{1}b_{0}\ \Lambda \ \ldots \ \Lambda \ a_{\ell
		}b_{\ell -1}\ \Lambda \ a_{\ell }b_{\ell }=a^{\prime }b^{\prime }$ be a
		canonical $\Lambda $-chain from $ab$ to $a^{\prime }b^{\prime }$. Again,
		note that all edges $a_{i}b_{i-1}$ and $a_{i}b_{i}$ of this $\Lambda $-chain
		have the same time-label $t$ in $(G,\lambda )$. We will prove by induction
		that $c^{\prime }a_{i}\in B$ and $b_{i}c^{\prime }\in A$ for every $%
		i=0,1,\ldots ,k$. First recall that $c^{\prime }a=c^{\prime }a_{0}\in B$.
		Furthermore, by applying part 1 of the proof to the triangle with vertices $%
		a_{0},b_{0},c$ and on the edge $c^{\prime }a_{0}\in B$, it follows that $%
		b_{0}c^{\prime }\in A$. This completes the induction basis.
		
		For the induction step, assume that $c^{\prime }a_{i}\in B$ and $%
		b_{i}c^{\prime }\in A$. If $a_{i+1}=a_{i}$ then clearly $c^{\prime
		}a_{i+1}\in B$. Suppose now that $a_{i+1}\neq a_{i}$. If $%
		\{a_{i+1},c^{\prime }\}\notin E$ then $a_{i+1}b_{i}\ \Lambda \ c^{\prime
		}b_{i}$. Then, since $a_{i+1}b_{i}\in C$ and $b_{i}c^{\prime }\in A$, it
		follows that $A=C^{-1}$, which is a contradiction to the assumption of the
		lemma. Therefore $\{a_{i+1},c^{\prime }\}\in E$. Now, since $%
		a_{i}b_{i}~\Lambda ~a_{i+1}b_{i}$, it follows that either $%
		\{a_{i},a_{i+1}\}\notin E$ or $\lambda (a_{i},a_{i+1}) < t$. In either
		case it follows that $c^{\prime }a_{i}\ \Lambda \ c^{\prime }a_{i+1}$.
		Therefore, since $c^{\prime }a_{i}\in B$, it follows that $c^{\prime
		}a_{i+1}\in B$.
		
		If $b_{i+1}=b_{i}$ then clearly $b_{i+1}c^{\prime }\in A$. Suppose now that $%
		b_{i+1}\neq b_{i}$. Then, since $c^{\prime }a_{i+1}\in B$, $a_{i+1}b_{i}\in
		C $, and $b_{i}c^{\prime }\in A$, we can apply part 1 of the lemma to the
		triangle with vertices $a_{i+1},b_{i},c^{\prime }$ and on the edge $%
		a_{i+1}b_{i+1}\in C$, from which it follows that $b_{i}c^{\prime }\in A$.
		This completes the induction step, and thus $c^{\prime }a_{k}=c^{\prime
		}a^{\prime }\in B$.
		
		\item Suppose that $ad\in A$ (resp.~$da\in A$), for some vertex $d$.
		Then, by setting $b^{\prime }=a$ and $c^{\prime }=d$ (resp.~$b^{\prime }=d$
		and $c^{\prime }=a$), part 1 of the lemma implies that $ab^{\prime }=aa\in C$
		(resp.~$c^{\prime }a=aa\in B$). Thus is a contradiction, as the underlying
		graph $G$ does not have the edge~$aa$.\qedhere
	\end{enumerate}
\end{proof}

\subparagraph{Deciding temporal transitivity using Boolean
	satisfiability.\label{boolean-transitivity-subsec}}

Starting with any undirected edge $\{u,v\}$ of the underlying graph $G$, we
can clearly enumerate in polynomial time the whole $\Lambda $-implication
class $A$ to which the oriented edge $uv$ belongs (cf.~\Cref{Lambda-def-eq}). If the reversely directed
edge $vu\in A$ then \Cref{reverse-implication-classes-lem} implies that 
$A=A^{-1}=\widehat{A}$. Otherwise, if $vu\notin A$ then $vu\in A^{-1}$ and
\Cref{reverse-implication-classes-lem} implies that $A\cap
A^{-1}=\emptyset $. Thus, we can also decide in polynomial time whether $%
A\cap A^{-1}=\emptyset $. If we encounter at least one $\Lambda $%
-implication class $A$ such that $A\cap A^{-1}\neq \emptyset $, then it
follows by \Cref{disjoint-implication-classes-lem} that $(G,\lambda )$
is not temporally transitively orientable. 

In the remainder of the section we will assume that $A\cap A^{-1}=\emptyset $
for every $\Lambda $-implication class $A$ of $(G,\lambda )$, which is a 
\emph{necessary} condition for $(G,\lambda )$ to be temporally transitive 
orientable. Moreover it follows by \Cref%
{disjoint-implication-classes-lem} that, if $(G,\lambda )$ admits a
temporally transitively orientation~$F$, then either $A\subseteq F$ or $%
A^{-1}\subseteq F$. This allows us to define a Boolean variable $x_{A}$ for
every $\Lambda $-implication class $A$, where $x_{A}=\overline{x_{A^{-1}}}$.
Here $x_{A}=1$ (resp.~$x_{A^{-1}}=1$) means that $A\subseteq F$ (resp.~$%
A^{-1}\subseteq F$), where $F$ is the temporally transitive orientation
which we are looking for. Let $\{A_{1},A_{2},\ldots ,A_{s}\}$ be a set of $%
\Lambda $-implication classes such that $\{\widehat{A_{1}},\widehat{A_{2}}%
,\ldots ,\widehat{A_{s}}\}$ 
is a partition of the edges of the underlying graph $G$.\footnote{Here we slightly abuse the notation by 
	identifying the undirected edge $\{u,v\}$ with the set of both its orientations $\{uv,vu\}$.} 
Then any truth assignment $\tau $ of the variables $%
x_{1},x_{2},\ldots ,x_{s}$ (where $x_i = x_{A_i}$ for every $i=1,2,\ldots,s$) corresponds bijectively to one possible orientation 
of the temporal graph $(G,\lambda )$, in which every $\Lambda $-implication
class is oriented consistently.

Now we define two Boolean formulas $\phi _{\text{3NAE}}$ and $\phi _{\text{%
		2SAT}}$ such that $(G,\lambda )$ admits a temporal transitive orientation if
and only if there is a truth assignment $\tau $ of the variables $%
x_{1},x_{2},\ldots ,x_{s}$ such that both $\phi _{\text{3NAE}}$ and $\phi _{%
	\text{2SAT}}$ are simultaneously satisfied. 
Intuitively, $\phi _{\text{3NAE}}$ captures the ``non-cyclic'' condition from \Cref{table:2sat-triangle} while $\phi _{\text{%
		2SAT}}$ captures the remaining conditions.
Here $\phi _{\text{3NAE}}$ is a
\textsc{3NAE} formula, i.e.,~the disjunction of clauses with three literals each,
where every clause $\nae(\ell _{1}, \ell _{2}, \ell _{3})$ is
satisfied if and only if at least one of the literals $\{\ell
_{1},\ell _{2},\ell _{3}\}$ is equal to 1 and at least one of them is equal
to 0. Furthermore $\phi _{\text{2SAT}}$ is a 2SAT formula, i.e.,~the
disjunction of 2CNF clauses with two literals each, where every clause $%
(\ell _{1}\vee \ell _{2})$ is satisfied if and only if at least one of the
literals $\{\ell _{1},\ell _{2}\}$ is equal to 1.

For simplicity of the
presentation we also define a variable $x_{uv}$ for every directed edge $uv$. 
More specifically, if $uv\in A_{i}$ (resp. $uv\in A_{i}^{-1}$) then we set 
$x_{uv}=x_{i}$ (resp. $x_{uv}=\overline{x_{i}}$). That is, $x_{uv}=\overline{%
	x_{vu}}$ for every undirected edge $\{u,v\}\in E$. Note that, although $%
\{x_{uv},x_{vu}:\{u,v\}\in E\}$ are defined as variables, they can
equivalently be seen as \emph{literals} in a Boolean formula over the
variables $x_{1},x_{2},\ldots ,x_{s}$. 
The process of building all $\Lambda$-implication classes and 
all variables $\{x_{uv}, x_{vu}:\{u,v\}\in E\}$ is given by
\Cref{edge-variables-alg}.

\begin{algorithm}[t]
	\caption{Building the $\Lambda$-implication classes and the edge-variables.}
	\label{edge-variables-alg}
	\begin{algorithmic}[1]
		\footnotesize
		\REQUIRE{A temporal graph $(G,\lambda)$, where $G=(V,E)$.}
		\ENSURE{The variables $\{x_{uv}, x_{vu}:\{u,v\}\in E\}$, or the announcement that $(G,\lambda)$ is temporally not transitively orientable.}
		
		\medskip
		
		\STATE{$s\leftarrow 0$; \ \ $E_0 \leftarrow E$} 
		\WHILE{$E_0 \neq \emptyset$}
		\STATE{$s\leftarrow s+1$; \ \ Let $\{p,q\}\in E_0$ be arbitrary}
		\STATE{Build the $\Lambda$-implication class $A_s$ of the oriented edge $pq$ (by~\Cref{Lambda-def-eq})}
		\IF[$A_s \cap A_s^{-1} \neq \emptyset$]{$qp\in A_s$}
		\RETURN{``NO''}
		\ELSE
		\STATE{$x_s$ is the variable corresponding to the directed edges of $A_s$}
		\FOR{every $uv\in A_s$}
		\STATE{$x_{uv}\leftarrow x_{s}$; \ \ $x_{vu}\leftarrow \overline{x_{s}}$}
		\COMMENT{$x_{uv}$ and $x_{vu}$ become aliases of $x_s$ and $\overline{x_s}$}
		\ENDFOR
		\STATE{$E_0 \leftarrow E_0 \setminus \widehat{A_s}$}
		\ENDIF
		\ENDWHILE
		\RETURN{$\Lambda$-implication classes $\{A_1, A_2, \ldots, A_s\}$ and variables $\{x_{uv}, x_{vu}:\{u,v\}\in E\}$}
	\end{algorithmic}
\end{algorithm}

\subparagraph{Description of the \textsc{3NAE} formula $\phi _{\text{3NAE}}$.} The formula $\phi _{\text{3NAE}}$ captures 
the ``non-cyclic'' condition of the problem variant \TTOs{} (presented in \Cref{table:2sat-triangle}). 
The formal description of $\phi _{\text{3NAE}}$ is as follows. 
Consider a synchronous triangle of $(G,\lambda )$ on the vertices $u,v,w$. 
Assume that $x_{uv}= x_{wv}$, i.e., $x_{uv}$ is the same variable as $x_{wv}$. 
Then the pair $\{uv,wv\}$ of oriented edges belongs 
to the same $\Lambda$-implication class $A_i$. 
This implies that the triangle on the vertices $u,v,w$ is never cyclically oriented in any proper orientation $F$ that respects $A_i$ or~$A_i^{-1}$. 
Note that, by symmetry, the same happens if $x_{vw}= x_{uw}$ or if $x_{wu}= x_{vu}$. 
Assume, on the contrary, that $x_{uv}\neq x_{wv}$, $x_{vw}\neq x_{uw}$, and $x_{wu}\neq x_{vu}$. 
In this case we add to $\phi _{\text{3NAE}}$ the clause $\nae(x_{uv},x_{vw}, x_{wu})$. 
Note that the triangle on $u,v,w$ is transitively oriented if and only if
$\nae(x_{uv}, x_{vw}, x_{wu})$ is satisfied, i.e.,~at least one of the
variables $\{x_{uv},x_{vw},x_{wu}\}$ receives the value 1 and at least one of them receives the value 0.

\subparagraph{Description of the 2SAT formula $\phi _{\text{2SAT}}$.} 
The formula $\phi _{\text{2SAT}}$ captures all conditions apart from the ``non-cyclic'' condition of the problem variant \TTOs{} (presented in \Cref{table:2sat-triangle}). 
The formal description of $\phi _{\text{2SAT}}$ is as follows. 
Consider a triangle of $(G,\lambda )$ on the vertices $u,v,w$,
where $\lambda (u,v)=t _{1}$, $\lambda (v,w)=t _{2}$, $\lambda
(w,v)=t _{3}$, and $t _{1}\leq t _{2}\leq t _{3}$. If $t
_{1}<t _{2}=t _{3}$ then we add to $\phi _{\text{2SAT}}$ the clauses $%
(x_{uw}\vee x_{wv})\wedge (x_{vw}\vee x_{wu})$; note that these clauses are
equivalent to $x_{wu}=x_{wv}$. If $t _{1}\leq t _{2}<t _{3}$ then
we add to $\phi _{\text{2SAT}}$ the clauses $(x_{wv}\vee x_{uw})\wedge
(x_{uv}\vee x_{wu})$; note that these clauses are equivalent to $%
(x_{vw}\Rightarrow x_{uw})\wedge (x_{vu}\Rightarrow x_{wu})$. Now consider a
path of length 2 that is induced by the vertices $u,v,w$, where $\lambda
(u,v)=t _{1}$, $\lambda (v,w)=t _{2}$, and $t _{1}\leq t _{2}$.
If $t _{1}=t _{2}$ then we add to $\phi _{\text{2SAT}}$ the clauses $%
(x_{vu}\vee x_{wv})\wedge (x_{vw}\vee x_{uv})$; note that these clauses are
equivalent to $(x_{uv}=x_{wv})$. Finally, if $t _{1}<t _{2}$ then we add
to $\phi _{\text{2SAT}}$ the clause $(x_{vu}\vee x_{wv})$; note that this
clause is equivalent to $(x_{uv}\Rightarrow x_{wv})$.

\bigskip In what follows, we say that $\phi _{\text{%
		3NAE}}\wedge \phi _{\text{2SAT}}$ is \emph{satisfiable} if and only if there
exists a truth assignment $\tau $ which simultaneously satisfies both $\phi _{\text{3NAE}}$ and $\phi _{\text{2SAT}}$. 
Given the above definitions of $\phi _{\text{3NAE}}$ and $\phi _{\text{2SAT}}$, it is easy
to check that their clauses model all conditions of the oriented edges
imposed by the row of ``\TTOs'' in~\Cref{table:2sat-triangle}.

\begin{observation}
	\label{transitivity-sat-first-equivalence-obs}
	The temporal graph $(G,\lambda )$ is transitively orientable if and only if 
	$\phi_{\text{3NAE}} \land \phi_{\text{2SAT}}$ is satisfiable.
\end{observation}

\newcommand{\InitalForcing}{\hyperref[initial-forcing-alg]{\textsc{Initial-Forcing}}}
\newcommand{\BooleanForcing}{\hyperref[phi-2sat-nae-forcing-alg]{\textsc{Boolean-Forcing}}}

Although deciding whether $\phi _{\text{2SAT}}$ is satisfiable can be done
in linear time with respect to the size of the formula~\cite{aspvall1979linear}, the
problem Not-All-Equal-3-SAT is NP-complete~\cite{Schaefer78}. 
We overcome this problem and present a polynomial-time algorithm for deciding 
whether $\phi_{\text{3NAE}}\wedge \phi_{\text{2SAT}}$ is satisfiable as follows.

\subparagraph{Roadmap of the entire process.}
Our algorithm iteratively produces at iteration $j$ a formula $\phi_{\text{3NAE}}^{(j)} \wedge \phi_{\text{2SAT}}^{(j)}$, which is 
computed from the previous formula $\phi_{\text{3NAE}}^{(j-1)} \wedge \phi_{\text{2SAT}}^{(j-1)}$ by (almost) simulating the classical greedy algorithm 
that solves 2SAT~\cite{aspvall1979linear}. 
The classical 2SAT-algorithm proceeds greedily as follows. For every variable $x_i$, if setting $x_i=1$ (resp.~$x_i=0$) leads to an immediate contradiction, 
the algorithm is forced to set $x_i=0$ (resp.~$x_i=1$). Otherwise, if each of the truth assignments $x_i=1$ and $x_i=0$ does not lead 
to an immediate contradiction, the algorithm arbitrarily chooses to set $x_i=1$ or $x_i=0$, and thus some clauses are removed from the formula 
as they were satisfied. 
The argument for the correctness of this classical 2SAT-algorithm is that new clauses are \emph{never added} to the formula at any step. 
The main technical difference between the 2SAT-algorithm and our algorithm is that, in our case, the formula $\phi_{\text{3NAE}}^{(j)} \wedge \phi_{\text{2SAT}}^{(j)}$ 
is \emph{not} necessarily a sub-formula of $\phi_{\text{3NAE}}^{(j-1)} \wedge \phi_{\text{2SAT}}^{(j-1)}$, as in some cases we need to also add clauses.

Our main technical result is that, nevertheless, 
if the algorithm does not return ``NO'' while applying variable \emph{forcings} at the initialization phase (during which $\phi_{\text{3NAE}}^{(0)} \wedge \phi_{\text{2SAT}}^{(0)}$ is computed), then the input instance is a \emph{yes}-instance. 
In this case, the algorithm proceeds by computing the formulas $\phi_{\text{3NAE}}^{(j)} \wedge \phi_{\text{2SAT}}^{(j)}$, for $j=1,2,\ldots$, which eventually determine a valid temporally transitive orientation of the input temporal graph. 
The proof of this result (see~\Cref{lemma:never-say-no} and~\Cref{thm-correctness-runningtime}) relies on a sequence of structural properties of temporal transitive orientations which we establish.
This phenomenon of deducing a polynomial-time algorithm for an algorithmic graph problem by deciding satisfiability of a mixed Boolean formula 
(i.e.~with both clauses of two and three literals) occurs rarely; this approach has been successfully used for the efficient recognition of 
simple-triangle (known also as ``PI'') graphs~\cite{Mertzios15}.

\subparagraph{Brief outline of the algorithm.} 
In the \emph{initialization phase}, we exhaustively check which truth values are \emph{forced} in $\phi _{\text{3NAE}} \wedge \phi _{\text{2SAT}}$ 
by using \InitalForcing{} (see \Cref{initial-forcing-alg}) as
a subroutine.
During the execution of \InitalForcing{}, we either replace the formulas $\phi _{\text{3NAE}}$ and $\phi _{\text{2SAT}}$ 
by the equivalent formulas $\phi_{\text{3NAE}}^{(0)}$ and $\phi_{\text{2SAT}}^{(0)}$, respectively, 
or we reach a contradiction by showing that $\phi _{\text{3NAE}} \wedge \phi _{\text{2SAT}}$ is unsatisfiable.

The \emph{main phase} of the algorithm starts once the formulas $\phi_{\text{3NAE}}^{(0)}$ and $\phi_{\text{2SAT}}^{(0)}$ have been computed. 
During this phase, we iteratively modify the formulas such that, 
at the end of iteration $j$ we have the formulas $\phi _{\text{3NAE}}^{(j)}$ and $\phi _{\text{2SAT}}^{(j)}$. 
Note that, during the execution of the algorithm, we can \emph{both add and remove} clauses from $\phi _{\text{2SAT}}^{(j)}$. 
On the other hand, we can \emph{only remove} clauses from $\phi_{\text{3NAE}}^{(j)}$. 
Thus, at some iteration $j$, we obtain $\phi_{\text{3NAE}}^{(j)}=\emptyset$, and after that iteration we only need to decide satisfiability of $\phi_{\text{2SAT}}^{(j)}$ which can be done efficiently~\cite{aspvall1979linear}.

\subparagraph{Two crucial technical lemmas.} 
For the remainder of the section we write $x_{ab}\overset{\ast }{\Rightarrow}_{\phi_{\text{2SAT}}}x_{uv}$ 
(resp.~$x_{ab}\overset{\ast }{\Rightarrow}_{\phi_{\text{2SAT}}^{(j)}}x_{uv}$) if the truth assignment $x_{ab}=1$ forces (in 0 or more iterations)
the truth assignment $x_{uv}=1$ from the clauses of $\phi_{\text{2SAT}}$ 
(resp.~of $\phi_{\text{2SAT}}^{(j)}$ at the iteration $j$ of the algorithm); 
in this case we say that $x_{ab}$ \emph{implies} $x_{uv}$ in $\phi_{\text{2SAT}}$ (resp.~in $\phi_{\text{2SAT}}^{(j)}$). 
We next introduce the notion of \emph{uncorrelated} triangles, which lets us
formulate some important properties of the implications in $\phi_{\text{2SAT}}$
and $\phi _{\text{2SAT}}^{(0)}$.
\begin{definition}
	\label{uncorellated-triangle-def}Let $u,v,w$ induce a synchronous triangle
	in $(G,\lambda )$, where each of the variables of the set $%
	\{x_{uv},x_{vu},x_{vw},x_{wv},x_{wu},x_{uw}\}$ belongs to a \emph{different} $%
	\Lambda $-implication class. If \emph{none} of the variables of the set $%
	\{x_{uv},x_{vu},x_{vw},x_{wv},x_{wu},x_{uw}\}$ implies any other variable of the same set 
	in the formula $\phi _{\text{2SAT}}$ (resp.~in the formula $\phi _{\text{2SAT}}^{(0)}$), 
	then the triangle of $u,v,w$ is \emph{$\phi _{\text{2SAT}}$-uncorrelated} (resp.~\emph{$\phi _{\text{2SAT}}^{(0)}$-uncorrelated}).
\end{definition}

Now we present our two crucial technical lemmas (\Cref{lemma-0,lemma-0-after-forcing}) which prove some structural properties of 
the 2SAT formulas $\phi _{\text{2SAT}}$ and $\phi _{\text{2SAT}}^{(0)}$. 
These structural properties will allow us to prove the correctness of our main algorithm 
in this section (\Cref{temporal-orientation-alg}). 
In a nutshell, these two lemmas guarantee that, whenever we have specific implications 
in $\phi _{\text{2SAT}}$ (resp.~in~$\phi _{\text{2SAT}}^{(0)}$), 
then we also have some specific \emph{other} implications in the same formula.

\begin{lemma}
	\label{lemma-0}Let $u,v,w$ induce a synchronous and $\phi _{\text{2SAT}}$-uncorrelated triangle in 
	$(G,\lambda )$, and let $\{a,b\}\in E$ be an edge of $G$ such that $%
	|\{a,b\} \cap \{u,v,w\}|\le 1$. If $x_{ab}\overset{\ast }{\Rightarrow
	}_{\phi _{\text{2SAT}}}x_{uv}$, then $x_{ab}$ also implies in $\phi _{\text{2SAT}}$ 
	at least one of the four variables in the set $\{x_{vw},x_{wv},x_{uw},x_{wu}\}$.
\end{lemma}

\begin{proof}
	Let $t$ be the common time-label of all the edges in the synchronous
	triangle of the vertices $u,v,w$. That is, $\lambda (u,v)=\lambda
	(v,w)=\lambda (w,u)=t$. Denote by $A$, $B$, and $C$ the $\Lambda $%
	-implication classes in which the directed edges $uv$, $vw$, and $%
	wu $ belong, respectively. Let $x_{ab}=x_{a_{0}b_{0}}\Rightarrow _{\phi_{\text{2SAT}%
	}}x_{a_{1}b_{1}}\Rightarrow _{\phi_{\text{2SAT}}}\ldots \Rightarrow _{\phi_{%
			\text{2SAT}}}x_{a_{k-1}b_{k-1}}\Rightarrow _{\phi_{\text{2SAT}%
	}}x_{a_{k}b_{k}}=x_{uv}$ be a $\phi_{\text{2SAT}}$-implication chain from $x_{ab}$ to $x_{uv}$. 
	Note that, without loss of generality, for each variable $x_{a_{i}b_{i}}$ in this chain, the
	directed edge $a_{i}b_{i}$ is a representative of a different $\Lambda $-implication class 
	than all other directed edges in the chain 
	(otherwise we can just shorten the $\phi_{\text{2SAT}}$-implication chain from $x_{ab}$ to $x_{uv}$). 
	Furthermore, since $x_{a_{k}b_{k}}=x_{uv}$, note that $a_{k}b_{k}$ and $uv$ are both
	representatives of the same $\Lambda $-implication class $A$. Therefore
	\Cref{triangle-lemma} (the temporal triangle lemma) implies that $%
	wa_{k}\in C$ and $b_{k}w\in B$. Therefore we can assume without loss of
	generality that $u=a_{k}$ and $v=b_{k}$. 
	Moreover, let $A'\notin \{A,A^{-1},B,B^{-1},C,C^{-1}\}$ be the $\Lambda$-implication class in which the directed edge $a_{k-1}b_{k-1}$ belongs. 
	Since $x_{a_{k-1}b_{k-1}}\Rightarrow _{\phi_{\text{2SAT}}}x_{a_{k}b_{k}}$, 
	note that without loss of generality we can choose the directed edge $a_{k-1}b_{k-1}$ to be such a representative of the $\Lambda$-implication 
	class $A'$ such that either $a_{k-1}=a_{k}$ or $b_{k-1}=b_{k}$. We now distinguish these two cases.

	\subparagraph{Case 1:} $u=a_{k}=a_{k-1}$ and $v=b_{k}\neq b_{k-1}$\textbf{.}
	Then, since $x_{a_{k-1}b_{k-1}}=x_{a_{k}b_{k-1}}\Rightarrow _{\phi_{\text{%
				2SAT}}}x_{a_{k}b_{k}}=x_{uv}$ and $\lambda (a_{k},b_{k})=t$, it follows that 
	$\lambda (u,b_{k-1})\geq t+1$. Suppose that $\{w,b_{k-1}\}\notin E$. Then $%
	x_{ub_{k-1}}\Rightarrow _{\phi_{\text{2SAT}}}x_{uw}$, which proves the
	lemma. Now suppose that $\{w,b_{k-1}\}\in E$. If $\lambda (w,b_{k-1})\leq
	\lambda (u,b_{k-1})-1$ then $x_{ub_{k-1}}\Rightarrow _{\phi_{\text{2SAT}%
	}}x_{uw}$, which proves the lemma. Suppose that $\lambda (w,b_{k-1})\geq
	\lambda (u,b_{k-1})+1 $. Then $x_{ub_{k-1}}\Rightarrow _{\phi_{\text{2SAT}%
	}}x_{wb_{k-1}}\Rightarrow _{\phi_{\text{2SAT}}}x_{wu}$, i.e.~$x_{ub_{k-1}}%
	\overset{\ast }{\Rightarrow }_{\phi_{\text{2SAT}}}x_{wu}$, which again proves
	the lemma. Suppose finally that $\lambda (w,b_{k-1})=\lambda (u,b_{k-1})$.
	Then, since $\lambda (u,w)=t<\lambda (w,b_{k-1})=\lambda (u,b_{k-1})$, it
	follows that $wb_{k-1}\ \Lambda \ ub_{k-1}$. If $\{v,b_{k-1}\}\notin E$ then 
	$x_{ub_{k-1}}=x_{wb_{k-1}}\Rightarrow _{\phi_{\text{2SAT}}}x_{wv}$, which
	proves the lemma. Now let $\{v,b_{k-1}\}\in E$. If $\lambda (v,b_{k-1})\leq
	\lambda (w,b_{k-1})-1$ then $x_{ub_{k-1}}=x_{wb_{k-1}}\Rightarrow _{\phi_{%
			\text{2SAT}}}x_{wv}$, which proves the lemma. If $\lambda (v,b_{k-1})\geq
	\lambda (w,b_{k-1})+1$ then $x_{ub_{k-1}}=x_{wb_{k-1}}\Rightarrow _{\phi_{%
			\text{2SAT}}}x_{vb_{k-1}}\Rightarrow _{\phi_{\text{2SAT}}}x_{wv}$, which
	proves the lemma. If $\lambda (v,b_{k-1})=\lambda (w,b_{k-1})$ then $%
	ub_{k-1}\ \Lambda \ vb_{k-1}$, and thus $x_{ub_{k-1}}=x_{a_{k-1}b_{k-1}}%
	\nRightarrow _{\phi_{\text{2SAT}}}x_{a_{k}b_{k}}=x_{uv}$, which is a
	contradiction.
	
	\subparagraph{Case 2:} $u=a_{k}\neq a_{k-1}$ and $v=b_{k}=b_{k-1}$\textbf{.}
	Then, since $x_{a_{k-1}b_{k-1}}=x_{a_{k-1}b_{k}}\Rightarrow _{\phi_{\text{%
				2SAT}}}x_{a_{k}b_{k}}=x_{uv}$ and $\lambda (a_{k},b_{k})=t$, it follows that 
	$\lambda (v,a_{k-1})\leq t-1$. Suppose that $\{w,a_{k-1}\}\notin E$. Then $%
	x_{a_{k-1}v}\Rightarrow _{\phi_{\text{2SAT}}}x_{wv}$, which proves the
	lemma. Now suppose that $\{w,a_{k-1}\}\in E$. If $\lambda (w,a_{k-1})\leq
	t-1 $ then $x_{a_{k-1}v}\Rightarrow _{\phi_{\text{2SAT}}}x_{wv}$, which
	proves the lemma. Suppose that $\lambda (w,a_{k-1})=t$. Then, since $\lambda
	(v,a_{k-1})\leq t-1$, it follows that $vw\ \Lambda \ a_{t-1}w$. If $%
	\{u,a_{k-1}\}\notin E$ then also $a_{t-1}w\ \Lambda \ uw$, and thus $%
	x_{wv}=x_{wu}$, which is a contradiction to the assumption that the triangle
	of $u,v,w$ is uncorrelated. Thus $\{u,a_{k-1}\}\in E$. If $\lambda
	(u,a_{k-1})\leq t-1$ then again $a_{k-1}w\ \Lambda \ uw$, which is a
	contradiction. On the other hand, if $\lambda (u,a_{k-1})\geq t$ then $%
	x_{a_{k-1}v}=x_{a_{k-1}b_{k-1}}\nRightarrow _{\phi_{\text{2SAT}%
	}}x_{a_{k}b_{k}}=x_{uv}$, which is a contradiction.
	
	Finally suppose that $\lambda (w,a_{k-1})\geq t+1$. Then, since $\lambda
	(v,w)=t$ and $\lambda (v,a_{k-1})\leq t-1$, it follows that $%
	x_{vw}\Rightarrow _{\phi_{\text{2SAT}}}x_{a_{k-1}w}\Rightarrow _{\phi_{\text{%
				2SAT}}}x_{a_{k-1}v}$. However, since $x_{a_{k-1}v}=x_{a_{k-1}b_{k}}%
	\Rightarrow _{\phi_{\text{2SAT}}}x_{a_{k}b_{k}}=x_{uv}$, it follows that $%
	x_{vw}\overset{\ast }{\Rightarrow }_{\phi_{\text{2SAT}}}x_{uv}$, which is a
	contradiction to the assumption that the triangle of $u,v,w$ is uncorrelated.
\end{proof}

\begin{lemma}
	\label{lemma-0-after-forcing}Let $u,v,w$ induce a synchronous and $\phi_{\text{2SAT}}^{(0)}$-uncorrelated triangle in $(G,\lambda )$, 
	and let $\{a,b\}\in E$ be an edge of $G$ such that $|\{a,b\} \cap \{u,v,w\}|\le 1$. 
	If $x_{ab}\overset{\ast}{\Rightarrow}_{\phi_{\text{2SAT}}^{(0)}}x_{uv}$, then
	$x_{ab}$ also implies in $\phi_{\text{2SAT}}^{(0)}$ 
	at least one of the four variables in the set $\{x_{vw},x_{wv},x_{uw},x_{wu}\}$. 
\end{lemma}

\begin{proof}
	Assume we have $|\{a,b\} \cap \{u,v,w\}|\le 1$ and
	$x_{ab}\overset{\ast}{\Rightarrow}_{\phi_{\text{2SAT}}^{(0)}}x_{uv}$. Then we
	make a case distinction on the last implication in the implication chain
	$x_{ab}\overset{\ast}{\Rightarrow}_{\phi_{\text{2SAT}}^{(0)}}x_{uv}$.
	\begin{enumerate}
		\item The last implication is an implication from $\phi_{\text{2SAT}}$, i.e.,
		$x_{ab}\overset{\ast}{\Rightarrow}_{\phi_{\text{2SAT}}^{(0)}}x_{pq}{\Rightarrow}_{\phi_{\text{2SAT}}}x_{uv}$.
		If $\{p,q\}\subseteq \{u,v,w\}$ then we are done, since we can assume that
		$\{p,q\}\neq \{u,v\}$ because no such implications are contained in
		$\phi_{\text{2SAT}}$. Otherwise \Cref{lemma-0} implies that $x_{pq}$ also implies at least
		one of the four variables in the set $\{x_{vw},x_{wv},x_{uw},x_{wu}\}$ in
		$\phi_{\text{2SAT}}$. If follows that $x_{ab}$ also implies at least
		one of the four variables in the set $\{x_{vw},x_{wv},x_{uw},x_{wu}\}$ in $\phi_{\text{2SAT}}^{(0)}$.
		\item The last implication is \emph{not} an implication from
		$\phi_{\text{2SAT}}$, i.e.,
		$x_{ab}\overset{\ast}{\Rightarrow}_{\phi_{\text{2SAT}}^{(0)}}x_{pq}{\Rightarrow}_{\phi_{\text{INIT}}}x_{uv}$,
		there the implication $x_{pq}{\Rightarrow}_{\phi_{\text{INIT}}}x_{uv}$ was
		added to $\phi_{\text{2SAT}}^{(0)}$ by \InitalForcing{}. If
		$x_{pq}{\Rightarrow}_{\phi_{\text{INIT}}}x_{uv}$ was added in Line~\ref{BF1} or
		Line~\ref{BF2} of \InitalForcing{}, then we have that
		$\{p,q\}\subseteq \{u,v,w\}$ and $\{p,q\}\neq \{u,v\}$, hence the $u,v,w$ is not
		a $\phi_{\text{2SAT}}^{(0)}$-uncorrelated triangle, a contradiction.
		If
		$x_{pq}{\Rightarrow}_{\phi_{\text{INIT}}}x_{uv}$ was added in Line~\ref{initadd} of \InitalForcing{}, then we have that
		$x_{pq}{\Rightarrow}_{\phi_{\text{INIT}}}x_{uw}$, hence we are done.
		\qedhere 
	\end{enumerate}
\end{proof}

\subparagraph{Detailed description of the algorithm.}
We are now ready to present our polynomial-time algorithm (Algorithm~\ref{temporal-orientation-alg}) for deciding whether a given temporal graph 
$(G,\lambda)$ is temporally transitively orientable. 
The main idea of our algorithm is as follows. 
First, the algorithm computes all $\Lambda$-implication classes $A_1,\ldots, A_s$ by calling~\Cref{edge-variables-alg} as a subroutine. 
If there exists at least one $\Lambda$-implication class $A_i$ where $uv,vu\in A_i$ for some edge $\{u,v\}\in E$, 
then we announce that $(G,\lambda)$ is a \no-instance, due
to~\Cref{disjoint-implication-classes-lem}.
Otherwise we associate to each $\Lambda$-implication class $A_i$ a variable $x_i$, 
and we build the \textsc{3NAE} formula $\phi_{\text{3NAE}}$ and the 2SAT formula $\phi_{\text{2SAT}}$, 
as described in~\Cref{boolean-transitivity-subsec}.

In the \emph{initialization phase} of Algorithm~\ref{temporal-orientation-alg}, 
we call algorithm \InitalForcing{} (see \Cref{initial-forcing-alg}) as a subroutine. 
Starting from the formulas $\phi_{\text{3NAE}}$ and $\phi_{\text{2SAT}}$, 
in \InitalForcing{} we build the formulas $\phi_{\text{3NAE}}^{(0)}$ and $\phi_{\text{2SAT}}^{(0)}$ by both 
(i)~checking which truth values are being \emph{forced} in $\phi _{\text{3NAE}} \wedge \phi _{\text{2SAT}}$ 
(lines~\ref{initial-first-forcing-line}-\ref{BF2}), 
and (ii)~adding to $\phi_{\text{2SAT}}$ some clauses that are implicitly implied in $\phi_{\text{3NAE}} \wedge \phi_{\text{2SAT}}$ 
(lines~\ref{initial-second-forcing-line}-\ref{initadd}). 
More specifically, \InitalForcing{} proceeds as follows:
(i) whenever setting $x_i=1$ (resp.~$x_i=0$) forces $\phi _{\text{3NAE}} \wedge \phi _{\text{2SAT}}$ 
to become unsatisfiable, we choose to set $x_i=0$ (resp.~$x_i=1$); 
(ii) if $x\Rightarrow_{\phi_{\text{2SAT}}^{(0)}}a$
and $x\Rightarrow_{\phi_{\text{2SAT}}^{(0)}}b$, and if we also have that
$\nae(a,b,c)\in\phi_{\text{3NAE}}^{(0)}$, then we add
$x\Rightarrow_{\phi_{\text{2SAT}}^{(0)}}\overline{c}$ to
$\phi_{\text{2SAT}}^{(0)}$, since clearly, if $x=1$ then $a=b=1$ and we have
to set $c=0$ to satisfy the NAE clause $\nae(a,b,c)$. 
The next observation follows easily by~\Cref{transitivity-sat-first-equivalence-obs} 
and by the construction of $\phi_{\text{3NAE}}^{(0)}$ and $\phi_{\text{2SAT}}^{(0)}$ in \InitalForcing{}.

\begin{observation}
	\label{transitivity-sat-second-equivalence-obs}
	The temporal graph $(G,\lambda )$ is
	transitively orientable if and only if $\phi_{\text{3NAE}}^{(0)} \land \phi_{\text{2SAT}}^{(0)}$ is satisfiable.
\end{observation}

The \emph{main phase} of the algorithm starts once the formulas $\phi_{\text{3NAE}}^{(0)}$ and $\phi_{\text{2SAT}}^{(0)}$ have been computed. 
As we prove in~\Cref{lemma:never-say-no}, if the algorithm does not conclude at the initialization phase that the input instance is a \emph{no}-instance, the the instance is a \emph{yes}-instance. 
During any iteration $j\geq 1$ of the algorithm, we pick an arbitrary variable $x_i$ and we assign it the truth value 1 (note that this is an arbitrary choice; we could equally choose to assign to $x_i$ the value~0). 
Once we have set $x_i=1$, we call algorithm \BooleanForcing{} (see \Cref{phi-2sat-nae-forcing-alg}) as a subroutine to check which
implications this value of $x_i$ has on the current formulas $\phi_{\text{3NAE}}^{(j-1)}$ and $\phi_{\text{2SAT}}^{(j-1)}$ and which other truth values of variables are forced. 
The correctness of \BooleanForcing{} can be easily verified by checking all subcases of \BooleanForcing{}.
During such a call of \BooleanForcing{} (i.e.~during an iteration $j\geq 1$ in the main phase of the algorithm), 
we replace the current formulas by $\phi_{\text{3NAE}}^{(j)}$ and $\phi_{\text{2SAT}}^{(j)}$, respectively. 
Summarizing, in its initialization phase, the algorithm decides whether the input temporal graph can be transitively oriented (i.e.~solves the decision version of the problem), while in its main phase it computes a temporally transitive orientation.

\begin{algorithm}[t!]
	\caption{\InitalForcing{}}
	\label{initial-forcing-alg}
	\begin{algorithmic}[1]
		\footnotesize
		\REQUIRE{A 2-SAT formula $\phi_{\text{2SAT}}$ and a 3-NAE formula $\phi_{\text{3NAE}}$\vspace{0,1cm}}
		\ENSURE{A 2-SAT formula $\phi_{\text{2SAT}}^{(0)}$ and a 3-NAE formula $\phi_{\text{3NAE}}^{(0)}$ 
			such that $\phi_{\text{2SAT}}^{(0)} \wedge \phi_{\text{3NAE}}^{(0)}$ is satisfiable 
			if and only if $\phi_{\text{2SAT}} \wedge \phi_{\text{3NAE}}$ is satisfiable, 
			or the announcement that $\phi_{\text{2SAT}} \wedge \phi_{\text{3NAE}}$ is not satisfiable.}
		
		\medskip
		
		\STATE{$\phi_{\text{3NAE}}^{(0)} \leftarrow \phi_{\text{3NAE}}$; \ \ $\phi_{\text{2SAT}}^{(0)} \leftarrow \phi_{\text{2SAT}}$} \COMMENT{initialization}
		
		\vspace{0,2cm}
		
		\FOR{every variable $x_i$ appearing in $\phi_{\text{3NAE}}^{(0)} \wedge \phi_{\text{2SAT}}^{(0)}$} \label{initial-first-forcing-line}
		
		\vspace{0,2cm}
		
		\IF{$\BooleanForcing{}\left(\phi_{\text{3NAE}}^{(0)},\phi_{\text{2SAT}}^{(0)}, x_i, 1\right) = \text{``NO''}$}
		
		\vspace{0,2cm}
		
		\IF{$\BooleanForcing{}\left(\phi_{\text{3NAE}}^{(0)},\phi_{\text{2SAT}}^{(0)}, x_i, 0\right) = \text{``NO''}$}
		
		\vspace{0,2cm}
		
		\RETURN{``NO''} \COMMENT{both $x_i=1$ and $x_i=0$ invalidate the formulas} \label{initial-return-no-line}
		
		\vspace{0,2cm}
		
		\ELSE 
		
		\vspace{0,2cm}
		
		\STATE{$\left(\phi_{\text{3NAE}}^{(0)},\phi_{\text{2SAT}}^{(0)}\right)
			\leftarrow
			\BooleanForcing{}\left(\phi_{\text{3NAE}}^{(0)},\phi_{\text{2SAT}}^{(0)},
			x_i, 0\right)$}\label{BF1}
		\ENDIF
		
		\vspace{0,2cm}
		
		\ELSE
		
		\vspace{0,2cm}
		
		\IF{$\BooleanForcing{}\left(\phi_{\text{3NAE}}^{(0)},\phi_{\text{2SAT}}^{(0)}, x_i, 0\right) = \text{``NO''}$}
		
		\vspace{0,2cm}
		
		\STATE{$\left(\phi_{\text{3NAE}}^{(0)},\phi_{\text{2SAT}}^{(0)}\right)
			\leftarrow
			\BooleanForcing{}\left(\phi_{\text{3NAE}}^{(0)},\phi_{\text{2SAT}}^{(0)},
			x_i, 1\right)$}\label{BF2}
		\ENDIF
		\ENDIF
		
		\ENDFOR
		
		\vspace{0,2cm}
		
		\FOR{every clause $\nae(x_{uv}, x_{vw}, x_{wu})$ of $\phi_{\text{3NAE}}^{(0)}$} \label{initial-second-forcing-line}
		\vspace{0,2cm}
		\FOR{every variable $x_{ab}$}
		\vspace{0,2cm}
		\IF[add $(x_{ab} \Rightarrow x_{uw})$ to $\phi_{\text{2SAT}}^{(0)}$]{$x_{ab} \overset{\ast }{\Rightarrow}_{\phi_{\text{2SAT}}^{(0)}} x_{uv}$ 
			and $x_{ab} \overset{\ast }{\Rightarrow}_{\phi_{\text{2SAT}}^{(0)}} x_{vw}$}
		\vspace{0,2cm}
		\STATE{$\phi_{\text{2SAT}}^{(0)} \leftarrow \phi_{\text{2SAT}}^{(0)} \wedge
			(x_{ba} \vee x_{uw})$}\label{initadd}
		\ENDIF
		\ENDFOR
		\ENDFOR

		\vspace{0,2cm}

		\STATE{Repeat Lines~\ref{initial-first-forcing-line} and~\ref{initial-second-forcing-line} until no changes occur on $\phi_{\text{2SAT}}^{(0)}$ and $\phi_{\text{3NAE}}^{(0)}$} 
		\label{initial-last-check-line}

		\vspace{0,2cm}
		
		\RETURN{$\left(\phi_{\text{3NAE}}^{(0)},\phi_{\text{2SAT}}^{(0)}\right)$}
		
	\end{algorithmic}
\end{algorithm}

\begin{algorithm}[t!]
	\caption{\BooleanForcing{}}
	\label{phi-2sat-nae-forcing-alg}
	\begin{algorithmic}[1]
		\footnotesize
		\REQUIRE{A 2-SAT formula $\phi_2$, a 3-NAE formula $\phi_3$, and a variable $x_i$ of $\phi_2 \wedge \phi_3$, and a truth value $\textsc{Value}\in \{0,1\}$}
		\ENSURE{A 2-SAT formula $\phi'_2$ and a 3-NAE formula $\phi'_3$, obtained from $\phi_2$ and $\phi_3$ by setting $x_i=\textsc{Value}$, or the announcement that $x_i=\textsc{Value}$ does not satisfy $\phi_2 \wedge \phi_3$.}
		
		\medskip
		
		\STATE{Let $a$ and $b$ be such that $x_{ab}=x_i$; \ \ $x_{ab} \leftarrow\textsc{Value}$}
		
		\medskip
		
		\STATE{$\phi'_2 \leftarrow \phi_2$; \ \ $\phi'_3 \leftarrow \phi_3$}
		
		\medskip
		
		\WHILE{$\phi'_2$ has a clause $(x_{uv} \vee x_{pq})$ and $x_{uv}=1$} \label{forcing-phi2-1-line}
		\STATE{Remove the clause $(x_{uv} \vee x_{pq})$ from $\phi'_2$}
		\ENDWHILE
		
		\medskip
		
		\WHILE{$\phi'_2$ has a clause $(x_{uv} \vee x_{pq})$ and $x_{uv}=0$} \label{forcing-phi2-2-line}
		\STATE{\textbf{if} $x_{pq}=0$ \textbf{then return} ``NO''}\label{forcing-phi2-2-b-line}
		\STATE{Remove the clause $(x_{uv} \vee x_{pq})$ from $\phi'_2$}\label{line:2satforce}
		\STATE{$x_{pq}	\leftarrow 1$; \ \ Repeat lines~\ref{forcing-phi2-1-line} and~\ref{forcing-phi2-2-line} until no changes occur in $\phi'_2$.}
		\COMMENT{Implement all changes to $\phi_2'$ that are implied by setting $x_{pq} = 1$}\label{line:clean1}
		
		\ENDWHILE

		\medskip
		
		\FOR[synchronous triangle on vertices $u,v,w$]{every clause $\nae(x_{uv}, x_{vw}, x_{wu})$ of $\phi'_3$} \label{forcing-phi3-line}
		
		\IF[add $(x_{uv} \Rightarrow x_{uw}) \wedge (x_{uw} \Rightarrow x_{vw})$ to $\phi'_2$]{$x_{uv} \overset{\ast }{\Rightarrow}_{\phi'_2} x_{vw}$} 
		\label{chain-length-phi3-line}
		\STATE{$\phi'_2 \leftarrow \phi'_2 \wedge (x_{vu} \vee x_{uw}) \wedge (x_{wu}
			\vee x_{vw})$}\label{line:add1}
		\STATE{Remove the clause $\nae(x_{uv},
			x_{vw}, x_{wu})$ from
			$\phi'_3$}\label{line:rmnae1}
		\ENDIF
		
		\vspace{0,2cm}
		
		\IF{$x_{uv}$ already got the value 1 or 0}
		\STATE{Remove the clause $\nae(x_{uv}, x_{vw}, x_{wu})$ from $\phi'_3$}\label{line:rmnae2}
		
		\vspace{0,2cm}
		
		\IF{$x_{vw}$ and $x_{wu}$ do not have yet a truth value}
		
		\vspace{0,2cm}
		
		\IF[add $(x_{vw}\Rightarrow x_{uw})$ to $\phi'_2$]{$x_{uv}=1$}
		\STATE{$\phi'_2 \leftarrow \phi'_2 \wedge (x_{wv} \vee x_{uw})$}\label{line:add2-1}
		\vspace{0,2cm}
		\ELSE[$x_{uv}=0$; \ in this case add $(x_{uw}\Rightarrow x_{vw})$ to $\phi'_2$]
		\STATE{$\phi'_2 \leftarrow \phi'_2 \wedge (x_{wu} \vee x_{vw})$}\label{line:add2-2}
		\ENDIF
		\ENDIF
		
		\vspace{0,2cm}
		
		\IF{$x_{vw}=x_{uv}$ and $x_{wu}$ does not have yet a truth value}
		\STATE{$x_{wu} \leftarrow 1-x_{uv}$; \ \ Repeat lines~\ref{forcing-phi2-1-line} and~\ref{forcing-phi2-2-line} until no changes occur in $\phi'_2$.}\label{line:add2-3}
		\COMMENT{Implement all changes to $\phi_2'$ that are implied by setting $x_{wu} = 1-x_{uv}$}\label{line:clean2}
		\ENDIF
		
		\vspace{0,2cm}
		
		\STATE{\textbf{if} $x_{vw}=x_{wu}=x_{uv}$ \textbf{then return} ``NO''} \label{line-all-equal-no}
		\ENDIF
		\ENDFOR
		
		\medskip
		
		\STATE{Repeat lines~\ref{forcing-phi2-1-line},~\ref{forcing-phi2-2-line}, and~\ref{forcing-phi3-line} until no changes occur in $\phi'_2$ and $\phi'_3$.}

		\medskip
		
		\RETURN{$(\phi'_2,\phi'_3)$}
		
	\end{algorithmic}
\end{algorithm}

\begin{algorithm}[t!]
	\caption{Temporal transitive orientation.}
	\label{temporal-orientation-alg}
	\begin{algorithmic}[1]
		\footnotesize
		\REQUIRE{A temporal graph $(G,\lambda)$, where $G=(V,E)$.}
		\ENSURE{A temporal transitive orientation $F$ of $(G,\lambda)$, or the announcement that $(G,\lambda)$ is temporally not transitively orientable.}
		
		\medskip
		
		\STATE{Execute Algorithm~\ref{edge-variables-alg} to build the $\Lambda$-implication classes $\{A_1, A_2, \ldots, A_s\}$ and the Boolean variables $\{x_{uv}, x_{vu}:\{u,v\}\in E\}$}
		\STATE{\textbf{if} Algorithm~\ref{edge-variables-alg} returns ``NO'' \textbf{then return} ``NO''}
		
		\vspace{0,2cm}
		
		\STATE{Build the \textsc{3NAE} formula $\phi_{\text{3NAE}}$ and the 2SAT formula $\phi_{\text{2SAT}}$}
		
		\vspace{0,2cm}
		
		\IF[Initialization phase]{$\InitalForcing{}\left(\phi_{\text{3NAE}},\phi_{\text{2SAT}}\right) \neq \text{``NO''}$}
		
		\vspace{0,2cm}
		
		\STATE{$\left(\phi_{\text{3NAE}}^{(0)},\phi_{\text{2SAT}}^{(0)}\right) \leftarrow \InitalForcing{}\left(\phi_{\text{3NAE}},\phi_{\text{2SAT}}\right)$} \label{force-initial-line}
		
		\vspace{0,2cm}
		
		\ELSE[$\phi_{\text{3NAE}} \wedge \phi_{\text{2SAT}}$ leads to a contradiction]
		\RETURN{``NO''} \label{initialize-return-no-line}
		\ENDIF
		
		\vspace{0,2cm}
		
		\STATE{$j\leftarrow 1$; \ \ $F\leftarrow \emptyset$} \COMMENT{Main phase}
		\WHILE[arbitrary choice of $x_i$]{a variable $x_{i}$ appearing in $\phi_{\text{3NAE}}^{(j-1)} \wedge \phi_{\text{2SAT}}^{(j-1)}$ did not yet receive a truth value}\label{line-star-1-main-alg}
		
		\vspace{0,2cm}

		\STATE{$\left(\phi_{\text{3NAE}}^{(j)},\phi_{\text{2SAT}}^{(j)}\right) \leftarrow \BooleanForcing{}\left(\phi_{\text{3NAE}}^{(j-1)},\phi_{\text{2SAT}}^{(j-1)}, x_i, 1\right)$} \label{force-1-line}
		
		\STATE{$j\leftarrow j+1$}\label{line-star-2-main-alg}
		\ENDWHILE
		
		\vspace{0,2cm}
		
		\FOR{$i=1$ to $s$}
		\STATE{\textbf{if} $x_i$ did not yet receive a truth value \textbf{then} $x_i \leftarrow 1$}\label{line-star-3-main-alg}
		\STATE{\textbf{if} $x_i=1$ \textbf{then} $F \leftarrow F \cup A_i$ \textbf{else} $F \leftarrow F \cup \overline{A_i}$}\label{line-star-4-main-alg}
		\ENDFOR
		
		\vspace{0,2cm}
		
		\RETURN{the temporally transitive orientation $F$ of $(G,\lambda)$}
	\end{algorithmic}
\end{algorithm}

\subparagraph{Correctness of the algorithm.}
We now formally prove that \Cref{temporal-orientation-alg} is correct. 
More specifically, we show that if \Cref{temporal-orientation-alg} gets a \yes-instance as input 
then it outputs a temporally transitive orientation, while if it gets a \no-instance as input then it outputs ``NO''.
The \emph{main result} of this section is~\Cref{thm-correctness-runningtime}, in which we prove that \TTO\ is correct and runs in polynomial time.

The next crucial observation follows immediately by the construction of $\phi_{\text{3NAE}}$ in~\Cref{boolean-transitivity-subsec}, and by the 
fact that, at every iteration $j$,~\Cref{temporal-orientation-alg} can only remove clauses from $\phi_{\text{3NAE}}^{(j-1)}$.

\begin{observation}\label{obs:naesat}
	When \BooleanForcing{} (\Cref{phi-2sat-nae-forcing-alg}) removes a clause from 
	$\phi_{\text{3NAE}}^{(j-1)}$, then this clause is satisfied by all satisfying
	assignments of $\phi_{\text{2SAT}}^{(j)}$.
\end{observation}

Next, we prove a crucial and involved technical lemma about the Boolean forcing steps of \Cref{temporal-orientation-alg}. 
This lemma will allow us to deduce that, during the \emph{main phase} of~\Cref{temporal-orientation-alg}, 
whenever a new clause is added to the 2SAT part of the formula, this happens only in lines~\ref{line:add2-1} and~\ref{line:add2-2} 
of \BooleanForcing{} (\Cref{phi-2sat-nae-forcing-alg}). 
That is, whenever a new clause is added to the 2SAT part of the formula in line~\ref{line:add1} of \Cref{phi-2sat-nae-forcing-alg}, 
this can only happen during the \emph{initialization phase} of~\Cref{temporal-orientation-alg}. 

\begin{lemma}\label{newsuperlemma}
	Consider an execution of \BooleanForcing{} (\Cref{phi-2sat-nae-forcing-alg}) which 
	is called in an iteration $j\geq 1$ (i.e.~in the main phase) of \Cref{temporal-orientation-alg}. 
	Then Lines~\ref{line:add1} and~\ref{line:rmnae1} of \BooleanForcing{} 
	are not executed.
\end{lemma}
\begin{proof}
	Assume for contradiction that Lines~\ref{line:add1} and~\ref{line:rmnae1} of
	\Cref{phi-2sat-nae-forcing-alg} are executed in iteration $j$ of
	\Cref{temporal-orientation-alg}.
	Let $j\geq 1$ be the first iteration where this happens. This means that there is a clause $\nae(x_{uv}, x_{vw}, x_{wu})$ of
	$\phi'_3$ and an implication $x_{uv} \overset{\ast }{\Rightarrow}_{\phi'_2}
	x_{vw}$ during the execution of \Cref{phi-2sat-nae-forcing-alg}.

	We first partition the implication chain $x_{uv} \overset{\ast }{\Rightarrow}_{\phi'_2}
	x_{vw}$ into ``old'' and ``new'' implications, where ``old'' implications are
	contained in $\phi_{\text{2SAT}}^{(0)}$ and all other implications (that were
	added in the previous iterations $1,2,\ldots,j-1$) are considered ``new''. 
	For simplicity of notation, we will refer to these ``new'' implications using the symbol ``$\Rightarrow_{\text{BF}}$''. 
	Recall here that, whenever $x_{ab}\Rightarrow_{\text{BF}} x_{cd}$, we have that $\lambda(a,b)=\lambda(c,d)$ by \BooleanForcing. 
	If there are several NAE clauses and 
	implication chains that fulfill the condition in Line~\ref{chain-length-phi3-line} of
	\Cref{phi-2sat-nae-forcing-alg}, we assume that $x_{uv} \overset{\ast }{\Rightarrow}_{\phi'_2}
	x_{vw}$ is one that contains a minimum number of ``new'' implications.
	Observe that, since we assume  $x_{uv} \overset{\ast }{\Rightarrow}_{\phi'_2}
	x_{vw}$ is a condition for the first execution of Lines~\ref{line:add1}
	and~\ref{line:rmnae1} of \Cref{phi-2sat-nae-forcing-alg}, it follows that all
	``new'' implications in $x_{uv} \overset{\ast }{\Rightarrow}_{\phi'_2}
	x_{vw}$ were added in Line~\ref{line:add2-1}
	or Line~\ref{line:add2-2} of \BooleanForcing\ (i.e.~\Cref{phi-2sat-nae-forcing-alg}) in previous
	iterations.

	Assume that $x_{uv} \overset{\ast }{\Rightarrow}_{\phi'_2} x_{vw}$ contains only ``old'' implications. 
	Then, this execution of Lines~\ref{line:add1} and~\ref{line:rmnae1} of \Cref{phi-2sat-nae-forcing-alg} happens during the initialization phase of \Cref{temporal-orientation-alg}. This is a contradiction to the assumption that this execution of Lines~\ref{line:add1} and~\ref{line:rmnae1} of \Cref{phi-2sat-nae-forcing-alg} happens at iteration $j\geq 1$ of \Cref{temporal-orientation-alg}.
	Therefore $x_{uv} \overset{\ast }{\Rightarrow}_{\phi'_2} x_{vw}$ contains at least one ``new'' implication. We now distinguish the cases where $x_{uv} \overset{\ast }{\Rightarrow}_{\phi'_2} x_{vw}$ contains ``old'' implications or not.

	\noindent\textbf{Case I: $x_{uv} \overset{\ast }{\Rightarrow}_{\phi'_2}
		x_{vw}$ contains at least one ``old'' implication.} We assume without loss of generality that $x_{uv} \overset{\ast }{\Rightarrow}_{\phi'_2}
	x_{vw}$ contains an ``old'' implication that is directly followed by a ``new''
	implication (if this is not the case, then we can consider the contraposition
	of the implication chain).
	
	Note that, since
	the ``new'' implication was added in Line~\ref{line:add2-1}
	or Line~\ref{line:add2-2} of \Cref{phi-2sat-nae-forcing-alg}, we can assume
	without loss of generality that the ``new'' implication is $x_{ab} {\Rightarrow}_{\text{BF}}
	x_{cb}$ and that $x_{ca}=1$ for some synchronous triangle on the vertices $a,b,c$ (this is the
	Line~\ref{line:add2-1} case, Line~\ref{line:add2-2} works analogously). 
	That is, we have $\nae(x_{ab},x_{bc},x_{ca})\in\phi_{\text{3NAE}}^{(0)}$. Let
	$x_{pq} {\Rightarrow}_{\phi_{\text{2SAT}}^{(0)}} x_{ab}$ be the ``old''
	implication. Then we have that $x_{pq} {\Rightarrow}_{\phi_{\text{2SAT}}^{(0)}}
	x_{ab} {\Rightarrow}_{\text{BF}} x_{cb}$ is contained in $x_{uv} \overset{\ast
	}{\Rightarrow}_{\phi'_2} x_{vw}$. Furthermore, by definition of
	$\phi_{\text{2SAT}}^{(0)}$, we have that $|\{p,q\}\cap\{a,b,c\}|\le 1$, hence
	we can apply \Cref{lemma-0-after-forcing} and obtain one of the following four scenarios:
	\begin{enumerate}
		\item $x_{pq} \overset{\ast }{\Rightarrow}_{\phi_{\text{2SAT}}^{(0)}} x_{cb}$:
		
		In this case we can replace $x_{pq}
		{\Rightarrow}_{\phi_{\text{2SAT}}^{(0)}} x_{ab} {\Rightarrow}_{\text{BF}}
		x_{cb}$ with $x_{pq} {\Rightarrow}_{\phi_{\text{2SAT}}^{(0)}} x_{cb}$ in the
		implication chain $x_{uv} \overset{\ast }{\Rightarrow}_{\phi_{\text{2SAT}}^{(j)}}
		x_{vw}$ to obtain an implication chain from $x_{uv}$ to $x_{vw}$ with strictly
		fewer ``new'' implications, a contradiction.
		\item $x_{pq} \overset{\ast }{\Rightarrow}_{\phi_{\text{2SAT}}^{(0)}} x_{bc}$:
		
		Now we have that $x_{pq} {\Rightarrow}_{\phi_{\text{2SAT}}^{(0)}}
		x_{ab}$ and $x_{pq} \overset{\ast }{\Rightarrow}_{\phi_{\text{2SAT}}^{(0)}} x_{bc}$. Then by
		definition of $\phi_{\text{2SAT}}^{(0)}$ we also have that $x_{pq}
		{\Rightarrow}_{\phi_{\text{2SAT}}^{(0)}} x_{ac}$. Recall that we have set $x_{ca}=1$, that is, $x_{ac}=0$. Therefore, by Lines~\ref{line:clean1} and~\ref{line:clean2} of \BooleanForcing,
		we have already set $x_{pq}=0$, and therefore the implication $x_{pq} {\Rightarrow}_{\phi_{\text{2SAT}}^{(0)}}
		x_{ab}$ does not exist in $\phi_2'$ anymore, which is a contradiction.
		\item $x_{pq} \overset{\ast }{\Rightarrow}_{\phi_{\text{2SAT}}^{(0)}} x_{ca}$:
		
		Now we have that $x_{pq} {\Rightarrow}_{\phi_{\text{2SAT}}^{(0)}}
		x_{ab}$ and $x_{pq} \overset{\ast }{\Rightarrow}_{\phi_{\text{2SAT}}^{(0)}} x_{ca}$. Then by
		definition of $\phi_{\text{2SAT}}^{(0)}$ we also have that $x_{pq}
		{\Rightarrow}_{\phi_{\text{2SAT}}^{(0)}} x_{cb}$. From here it is the same as Case 1.
		\item $x_{pq} \overset{\ast }{\Rightarrow}_{\phi_{\text{2SAT}}^{(0)}} x_{ac}$:   
		Same as Case 2.
	\end{enumerate}
	Hence, we have a contradiction in every case and can conclude that $x_{uv} \overset{\ast }{\Rightarrow}_{\phi'_2}
	x_{vw}$ does not contain any ``old'' implications.

	\medskip
	
	\noindent\textbf{Case II: $x_{uv} \overset{\ast }{\Rightarrow}_{\phi'_2} x_{vw}$ contains only ``new'' implications.} 
	To analyze this case, we first introduce the notion of \emph{alternating} and \emph{non-alternating} sequences of ``new'' implications, as follows. 
	Whenever the sequence $x_{uv} \overset{\ast }{\Rightarrow}_{\text{BF}} x_{vw}$ 
	contains at least one pair of consecutive direct implications of the form $x_{ab} {\Rightarrow}_{\text{BF}} x_{ac} {\Rightarrow}_{\text{BF}} x_{ad}$ (see~Figure~\ref{fig:alternating-non-alternating}(a)), 
	or of the form $x_{ba} {\Rightarrow}_{\text{BF}} x_{ca} {\Rightarrow}_{\text{BF}} x_{da}$ (see~Figure~\ref{fig:alternating-non-alternating}(b)), 
	we call $x_{uv} \overset{\ast }{\Rightarrow}_{\text{BF}} x_{vw}$ a \emph{non-alternating} sequence of implications; otherwise we call it \emph{alternating} 
	(see~Figure~\ref{fig:alternating-non-alternating}(c)). That is, if $x_{uv} \overset{\ast }{\Rightarrow}_{\text{BF}} x_{vw}$ is \emph{alternating}, then it either has the form
	\begin{equation}
		x_{uv} = x_{u_{1}v_{1}} {\Rightarrow}_{\text{BF}} 
		x_{u_{2}v_{1}} {\Rightarrow}_{\text{BF}} 
		x_{u_{2}v_{2}} {\Rightarrow}_{\text{BF}} 
		x_{u_{3}v_{2}} \overset{\ast }{\Rightarrow}_{\text{BF}} 
		x_{u_{j}v_{i}} = x_{vw},
		\label{alternating-sequence-1}
	\end{equation}
	or it has the form
	\begin{equation}
		x_{uv} = x_{u_{1}v_{1}} {\Rightarrow}_{\text{BF}} 
		x_{u_{1}v_{2}} {\Rightarrow}_{\text{BF}} 
		x_{u_{2}v_{2}} {\Rightarrow}_{\text{BF}} 
		x_{u_{2}v_{3}} \overset{\ast }{\Rightarrow}_{\text{BF}} 
		x_{u_{i}v_{j}} = x_{vw},
		\label{alternating-sequence-2}
	\end{equation}
	where either $j=i$ or $j=i+1$. 
	Figure~\ref{fig:alternating-non-alternating} illustrates some examples of alternating and non-alternating sequences of implications.

	\begin{figure}[t]
		\centering
		\begin{tikzpicture}[yscale=1.5]
			\begin{scope}[every node/.style={vertex}]
				\node[label=left:$a$] (u) at (0, 0) {};
				\node[label=right:$c$] (w) at (2, 0) {};
				\node[label=above:$b$] (v) at (1, 1) {};
				\node[label=below:$d$] (d) at (1, -1) {};
			\end{scope}
			\node at (1,-1.6) {(a)};
			\draw[edge,dash dot, bend left, green!60!black] (v) to node {} (d);
			\draw[edge] 
			(v)
			-- node[timelabel] {$t$} (u)
			(u)
			-- node[timelabel] {$t$} (w)
			(u)
			-- node[timelabel] {$t$} (d);	
			\draw[diredge,red,line width=1.5pt] 
			(v)
			-- node[timelabel] {$t$} (w);		
			\draw[diredge,red,line width=1.5pt] 
			(w)
			-- node[timelabel] {$t$} (d);

			\begin{scope}[every node/.style={vertex}]
				\node[label=left:$a$] (u2) at (4, 0) {};
				\node[label=right:$c$] (w2) at (6, 0) {};
				\node[label=above:$b$] (v2) at (5, 1) {};
				\node[label=below:$d$] (d2) at (5, -1) {};
			\end{scope}
			\node at (5,-1.6) {(b)};
			\draw[edge,dash dot, bend left, green!60!black] (v2) to node[] {}
			(d2);
			\draw[edge] 
			(v2)
			-- node[timelabel] {$t$} (u2)
			(u2)
			-- node[timelabel] {$t$} (w2)
			(u2)
			-- node[timelabel] {$t$} (d2);	
			\draw[diredge,red,line width=1.5pt] 
			(w2)
			-- node[timelabel] {$t$} (v2);		
			\draw[diredge,red,line width=1.5pt] 
			(d2)
			-- node[timelabel] {$t$} (w2);

			\begin{scope}[every node/.style={vertex}]
				\node[label=left:$a$] (u3) at (8, 0) {};
				\node[label=right:$c$] (w3) at (10, 0) {};
				\node[label=above:$b$] (v3) at (9, 1) {};
				\node[label=below:$d$] (d3) at (9, -1) {};
			\end{scope}
			\node at (9,-1.6) {(c)};
			\draw[edge,specialedge,dash dot, bend left, green!60!black] (v3) to node {}
			(d3);
			\draw[edge] 
			(v3)
			-- node[timelabel] {$t$} (u3)
			(w3)
			-- node[timelabel] {$t$} (d3)
			(w3)
			-- node[timelabel] {$t$} (u3);	
			\draw[diredge,red,line width=1.5pt] 
			(v3)
			-- node[timelabel] {$t$} (w3);		
			\draw[diredge,red,line width=1.5pt] 
			(d3)
			-- node[timelabel] {$t$} (u3);	
		\end{tikzpicture}
		\caption{Illustration of alternating and non-alternating sequences of implications that can occur at the two synchronous triangles on the vertices $\{a,b,c\}$ and $\{a,c,d\}$. The red directed edges illustrate variables that have already been set to $1$ by the algorithm \BooleanForcing. 
			Figure~(a): \emph{non-alternating} implications 
			$x_{ab} {\Rightarrow}_{\text{BF}} x_{ac} {\Rightarrow}_{\text{BF}} x_{ad}$, which occur whenever $x_{bc}=x_{cd}=1$ (red edges). 
			Figure~(b): \emph{non-alternating} implications 
			$x_{ba} {\Rightarrow}_{\text{BF}} x_{ca} {\Rightarrow}_{\text{BF}} x_{da}$, which occur whenever $x_{cb}=x_{dc}=1$ (red edges). 
			Figure~(c): \emph{alternating} implications 
			$x_{ab} {\Rightarrow}_{\text{BF}} x_{ac} {\Rightarrow}_{\text{BF}} x_{dc}$, which occur whenever $x_{bc}=x_{da}=1$ (red edges). 
			In all three figures, the green
			dash-dotted line indicates that edge $\{a,d\}$ may exist (with some time label) or may not exist.}
		\label{fig:alternating-non-alternating}
	\end{figure}

	We now distinguish the cases where $x_{uv} \overset{\ast }{\Rightarrow}_{\text{BF}} x_{vw}$ is an alternating or a non-alternating sequence of implications. 
	Note that, as all these are ``new'' implications, 
	all edges which are involved in $x_{uv} \overset{\ast }{\Rightarrow}_{\text{BF}} x_{vw}$ have the same label $t$. 
	That is, for every variable $x_{ab}$ that appears in the sequence $x_{uv} \overset{\ast }{\Rightarrow}_{\text{BF}} x_{vw}$ of implications, we have that $\lambda(a,b)=t$.
	
	\medskip
	
	\noindent\textbf{Case II-A: $x_{uv} \overset{\ast }{\Rightarrow}_{\text{BF}} x_{vw}$ is a non-alternating sequence of implications.} 
	Without loss of generality, let this sequence $x_{uv} \overset{\ast }{\Rightarrow}_{\text{BF}} x_{vw}$ contain the pair of consecutive direct implications 
	$x_{ab} {\Rightarrow}_{\text{BF}} x_{ac} {\Rightarrow}_{\text{BF}} x_{ad}$ 
	(the case where it contains the implications $x_{ba} {\Rightarrow}_{\text{BF}} x_{ca} {\Rightarrow}_{\text{BF}} x_{da}$ can be treated in exactly the same way). 
	
	Let $a,b,c$ be the vertices of the synchronous triangle that caused the implication $x_{ab} {\Rightarrow}_{\text{BF}} x_{ac}$, 
	and let $a',c',d$ be the vertices of the synchronous triangle that caused the implication $x_{ac} {\Rightarrow}_{\text{BF}} x_{ad}$, where $x_{ac}=x_{a'c'}$ and $x_{ad}=x_{a'd}$. Then, 
	\Cref{triangle-lemma} (the temporal triangle lemma) implies that
	the edges $\{a,d\}$ and $\{c,d\}$ exist in the graph and that $ad$ (resp.~$cd$) belongs to the same $\Lambda$-implication class with $a'd$ (resp.~$c'd$). 
	Therefore we can assume without loss of generality that $a=a'$ and $c=c'$.

	Then, since $x_{ab} {\Rightarrow}_{\text{BF}} x_{ac}$ and $x_{ac} {\Rightarrow}_{\text{BF}} x_{ad}$ are direct ``new'' implications, it follows that $x_{bc}=x_{cd}=1$ (as these implications have only been added by Lines~\ref{line:add2-1}
	or~\ref{line:add2-2} of \BooleanForcing).

	Let $\{b,d\}\notin E$ or $\lambda(b,d)<t$. Then 
	$\phi_{\text{2SAT}}^{(0)}$ by definition contains
	$x_{ab} {\Rightarrow}_{\phi_{\text{2SAT}}^{(0)}} x_{ad}$. 
	Thus, we can replace within $x_{uv} \overset{\ast }{\Rightarrow}_{\text{BF}} x_{vw}$ the two ``new'' implications 
	$x_{ab} {\Rightarrow}_{\text{BF}} x_{ac} {\Rightarrow}_{\text{BF}} x_{ad}$ 
	by the ``old'' implication $x_{ab} {\Rightarrow}_{\phi_{\text{2SAT}}^{(0)}} x_{ad}$, 
	thus resulting to a a sequence of implications from $x_{uv}$ to~$x_{vw}$ that has fewer ``new'' implications, a contradiction to our assumption. 
	
	Let $\lambda(b,d)>t$. Then $\phi_{\text{2SAT}}^{(0)}$ by definition contains
	$x_{cd} {\Rightarrow}_{\phi_{\text{2SAT}}^{(0)}} x_{bd}$ and 
	$x_{bd} {\Rightarrow}_{\phi_{\text{2SAT}}^{(0)}} x_{ba}$. 
	Thus, since $x_{cd}=1$, it follows \BooleanForcing\ sets $x_{ab}=0$, which is a contradiction to the assumption that the implication $x_{ab} {\Rightarrow}_{\text{BF}} x_{ac}$ belongs to $\phi_{2}'$.
	
	Let now $\lambda(b,d)=t$. Then $\nae(x_{bc},x_{cd},x_{db})\in\phi_{\text{3NAE}}^{(0)}$. 
	If $x_{bc}$ is set to $1$ before $x_{cd}$ is set to $1$ (i.e.~at an earlier iteration of \BooleanForcing), then \BooleanForcing\ adds (in Line~\ref{line:add2-1}) to $\phi_{2}'$ the direct implication $x_{cd} {\Rightarrow}_{\text{BF}} x_{bd}$. 
	In this case, when $x_{cd}$ is set to $1$ at a subsequent iteration of \BooleanForcing, $x_{bd}$ is also set to $1$. 
	Similarly, if $x_{cd}$ is set to $1$ before $x_{bc}$ is set to $1$, then \BooleanForcing\ adds to $\phi_{2}'$ the direct implication $x_{db} {\Rightarrow}_{\text{BF}} x_{cb}$, 
	which is equivalent to $x_{bc} {\Rightarrow}_{\text{BF}} x_{bd}$. In this case, when $x_{bd}$ is set to $1$ at a subsequent iteration of \BooleanForcing, $x_{bd}$ is also set to $1$. 
	Finally, if both $x_{bc}$ and $x_{cd}$ are set to $1$ at the same iteration, \BooleanForcing\ also sets $x_{bd}$ to $1$ in Line~\ref{line:add2-3}. 
	Summarizing, in any case \BooleanForcing\ always sets $x_{bd}=1$, and thus it also adds to $\phi_{2}'$ the implication $x_{ab} {\Rightarrow}_{\text{BF}} x_{ad}$. 
	Thus, we can replace within $x_{uv} \overset{\ast }{\Rightarrow}_{\text{BF}} x_{vw}$ the two implications 
	$x_{ab} {\Rightarrow}_{\text{BF}} x_{ac} {\Rightarrow}_{\text{BF}} x_{ad}$ 
	by the single implication $x_{ab} {\Rightarrow}_{\text{BF}} x_{ad}$, 
	thus resulting to a sequence of implications from $x_{uv}$ to~$x_{vw}$ that has fewer ``new'' implications, a contradiction to our assumption. 
	
	\medskip
	
	\noindent\textbf{Case II-B: $x_{uv} \overset{\ast }{\Rightarrow}_{\text{BF}} x_{vw}$ is an alternating sequence of implications.} 
	Let this sequence be of the form of (\ref{alternating-sequence-1}) where $j=i$ (the cases where $j=i+1$ or where the sequence is of the form of (\ref{alternating-sequence-2}) can be treated analogously), that is,
	\begin{equation}
		x_{uv} = x_{u_{1}v_{1}} {\Rightarrow}_{\text{BF}} 
		x_{u_{2}v_{1}} {\Rightarrow}_{\text{BF}} 
		x_{u_{2}v_{2}} {\Rightarrow}_{\text{BF}} 
		x_{u_{3}v_{2}} \overset{\ast }{\Rightarrow}_{\text{BF}} 
		x_{u_{i}v_{i}} = x_{vw},
		\label{alternating-sequence-3}
	\end{equation}
	
	Similarly to Case II-A, by iteratively applying \Cref{triangle-lemma} (the temporal triangle lemma), we may assume without loss of generality that all implications of (\ref{alternating-sequence-3}) are added to $\phi_2'$ by the synchronous triangles on the vertices 
	$\{u_1,v_1,u_2\}$, 
	$\{v_1,u_2,v_2\}$, 
	$\{u_2,v_2,u_3\}$, $\ldots$, 
	$\{v_{i-1},u_i,v_i\}$. 
	Furthermore, as all the implications of (\ref{alternating-sequence-3}) have been added to $\phi_2'$ by \BooleanForcing, it follows that $x_{u_{i}u_{i-1}}=x_{u_{i-1}u_{i-2}}=\ldots=x_{u_{2}u_{1}}=1$ and $x_{v_{1}v_{2}}=x_{v_{2}v_{3}}=\ldots=x_{v_{i-1}v_{i}}=1$. 
	
	Now, since $x_{u_{i}v_{i}} = x_{vw}$ (i.e.~$u_{i}v_{i}$ belongs to the same $\Lambda$-implication class with $vw$), it follows by 
	\Cref{triangle-lemma} (the temporal triangle lemma) that 
	the edge $\{u_{1},u_{i}\}$ exists in the graph and that $u_{1}u_{i}$ belongs to the same $\Lambda$-implication class with $u_{1}v = uv$ (and thus, in particular, $\lambda(u_{1},u_{i}) = \lambda(u_{1},v) = t$).

	Recall that $\lambda(u_{1},u_{2})= t$ and $x_{u_{2}u_{1}}=1$. 
	We now prove by induction that, for every $j=3,\ldots,i$, we have $\lambda(u_{1},u_{j})\geq t$ and $x_{u_{j}u_{1}}=1$. 
	
	For the induction basis, let $j=3$.
	If $\{u_{1},u_{3}\} \notin E$ or $\lambda(u_{1},u_{3})<t$, then $\phi_{\text{2SAT}}^{(0)}$ by definition contains
	$x_{u_{3}u_{2}} {\Rightarrow}_{\phi_{\text{2SAT}}^{(0)}} x_{u_{1}u_{2}}$. This is a contradiction, as $x_{u_{3}u_{2}}=x_{u_{2}u_{1}}=1$. 
	Therefore $\{u_{1},u_{3}\} \in E$ and $\lambda(u_{1},u_{3})\geq t$. 
	If $\lambda(u_{1},u_{3})= t$ then \BooleanForcing\ sets $x_{u_{3}u_{1}}=1$ (see Line~\ref{line:add2-3} of \BooleanForcing). 
	If $\lambda(u_{1}u_{3})> t$ then $\phi_{\text{2SAT}}^{(0)}$ contains 
	$x_{u_{2}u_{1}} {\Rightarrow}_{\phi_{\text{2SAT}}^{(0)}} x_{u_{3}u_{1}}$. Therefore, since $x_{u_{2}u_{1}}=1$, it follows in this case as well that \BooleanForcing\ sets $x_{u_{3}u_{1}}=1$. This completes the induction basis.
	
	For the induction step, let $4\leq j \leq i$, and assume by the induction hypothesis that 
	$t'=\lambda(u_{1},u_{j-1})\geq t$ and $x_{u_{j-1}u_{1}}=1$. Recall that $\lambda(u_{j-1},u_{j})= t$ and $x_{u_{j}u_{j-1}}=1$. 
	If $\{u_{1},u_{j}\} \notin E$ or $\lambda(u_{1},u_{j})<t'$, then $\phi_{\text{2SAT}}^{(0)}$ by definition contains
	$x_{u_{j}u_{j-1}} {\Rightarrow}_{\phi_{\text{2SAT}}^{(0)}} x_{u_{1}u_{j-1}}$. This is a contradiction, as $x_{u_{j}u_{j-1}}=x_{u_{j-1}u_{1}}=1$. 
	Therefore $\{u_{1},u_{j}\} \in E$ and $\lambda(u_{1},u_{j})\geq t'$. 
	If $\lambda(u_{1},u_{j})= t'=t$ then \BooleanForcing\ sets $x_{u_{j}u_{1}}=1$ (see Line~\ref{line:add2-3} of \BooleanForcing). 
	If $\lambda(u_{1},u_{j})= t'>t$ 
	or if $\lambda(u_{1},u_{j})> t'\geq t$ 
	then $\phi_{\text{2SAT}}^{(0)}$ contains 
	$x_{u_{j-1}u_{1}} {\Rightarrow}_{\phi_{\text{2SAT}}^{(0)}} x_{u_{j}u_{1}}$. Therefore, since $x_{u_{j-1}u_{1}}=1$, it follows in this case as well that \BooleanForcing\ sets $x_{u_{j}u_{1}}=1$. 
	This completes the induction step.

	Therefore, in particular, for $j=i$ we have that 
	$x_{u_{i}u_{1}}=1$. Thus, since $u_{1}u_{i}$ belongs to the same $\Lambda$-implication class with $u_{1}v = uv$, it follows that $x_{uv}=1$, which is a contradiction to the assumption that $x_{uv} \overset{\ast }{\Rightarrow}_{\text{BF}} x_{vw}$ is contained in $\phi_{2}'$. 
	This completes the proof.
\end{proof}

In the next lemma we prove that, if \Cref{temporal-orientation-alg} does not return ``NO'' after the initialization phase (in Line~\ref{initialize-return-no-line}), then the 2SAT formula $\phi_{\text{2SAT}}^{(0)}$ is satisfiable. 
Furthermore, as we prove in~\Cref{super2satlemma}, in this case also the 2SAT formulas $\phi_{\text{2SAT}}^{(j)}$ are satisfiable for every~$j\geq 1$.

\begin{lemma}\label{phi-2-0-satisfiable}
	Assume that \Cref{temporal-orientation-alg} does not return ``NO'' in the initialization phase (i.e.~in Line~\ref{initialize-return-no-line}). 
	Then there exists no variable $x_{uv}$ in $\phi_{\text{2SAT}}^{(0)}$ such that $x_{uv} \overset{\ast }{\Rightarrow}_{\phi_{\text{2SAT}}^{(0)}} x_{vu}$, 
	and thus $\phi_{\text{2SAT}}^{(0)}$ is satisfiable.
\end{lemma}

\begin{proof}
	Since \Cref{temporal-orientation-alg} does not return ``NO'' in Line~\ref{initialize-return-no-line}, 
	it follows that Line~\ref{initial-return-no-line} of~\InitalForcing{} (\Cref{initial-forcing-alg}) is not executed, when~\InitalForcing{} is called by \Cref{temporal-orientation-alg}. 
	Furthermore, before~\InitalForcing{} finishes, it checks in Line~\ref{initial-last-check-line} whether any of the formulas $\phi_{\text{3NAE}}^{(0)}$ or $\phi_{\text{2SAT}}^{(0)}$ 
	have been changed since the last iteration of Lines~\ref{initial-first-forcing-line} and~\ref{initial-second-forcing-line}. 
	
	Let $x_{uv}$ be an arbitrary variable in $\phi_{\text{2SAT}}^{(0)}$, i.e.~in the 2SAT part of the formula after~\InitalForcing{} has finished. 
	Since $x_{uv}$ did not get a Boolean value during the execution of~\InitalForcing{}, it follows that, when~\InitalForcing{} stops, 
	setting $x_{uv}$ to 1 (resp.~to 0) does not cause a contradiction. Indeed, otherwise~\InitalForcing{} would set $x_{uv}$ equal to 0 (resp.~1). 
	Therefore, once~\InitalForcing{} finishes, there cannot exist any variable $x_{uv}$ in $\phi_{\text{2SAT}}^{(0)}$ such that $x_{uv} \overset{\ast }{\Rightarrow}_{\phi_{\text{2SAT}}^{(0)}} x_{vu}$ (as otherwise~\InitalForcing{} would set $x_{uv}=0$). 
	This completes the lemma.
\end{proof}

\begin{lemma}\label{super2satlemma}
	Assume that \Cref{temporal-orientation-alg} does not return ``NO'' in the initialization phase (i.e.~in Line~\ref{initialize-return-no-line}). 
	Then, at any point during an arbitrary call of \BooleanForcing\ at the iteration $j\ge 1$ of \Cref{temporal-orientation-alg}, 
	there does not exist any variable $x_{uv}$ in $\phi_2'$ such that $x_{uv} \overset{\ast }{\Rightarrow}_{\phi_2'} x_{vu}$, and thus $\phi_2'$ is satisfiable.
\end{lemma}

\begin{proof}
	Let $j=1$. At the very beginning of iteration $j=1$ (where no changes have been made to $\phi_2'$ by \BooleanForcing) it follows immediately by~\Cref{phi-2-0-satisfiable} that there is no variable $x_{uv}$ in $\phi_2'=\phi_{\text{2SAT}}^{(0)}$ such that $x_{uv} \overset{\ast }{\Rightarrow}_{\phi_2'} x_{vu}$.

	Now, let $j\ge 1$. 
	Assume that, at the very beginning of iteration $j$, there is no variable $x_{uv}$ in $\phi_2'=\phi_{\text{2SAT}}^{(j-1)}$ such that $x_{uv} \overset{\ast }{\Rightarrow}_{\phi_2'} x_{vu}$.
	For the sake of contradiction, assume that, at some point during the execution of this call of \BooleanForcing, 
	there exists a variable $x_{uv}$ in $\phi_2'$ such that $x_{uv} \overset{\ast }{\Rightarrow}_{\phi_2'} x_{vu}$. 
	Assume that this is the earliest point during the execution of this call of \BooleanForcing\ where such an implication chain $x_{uv} \overset{\ast }{\Rightarrow}_{\phi_2'} x_{vu}$ exists in $\phi_2'$. Furthermore, among all implication chains $x_{uv} \overset{\ast }{\Rightarrow}_{\phi_2'} x_{vu}$, consider one that has the smallest number of ``new'' implications.

	Similarly to the proof of~\Cref{newsuperlemma}, we partition the implication chain $x_{uv} \overset{\ast }{\Rightarrow}_{\phi_2'} x_{vu}$ into ``old'' implications 
	(which are also present in $\phi_{\text{2SAT}}^{(0)}$) and ``new'' implications 
	(which were added by \BooleanForcing{} during some iteration $j'\in\{1,2,\ldots,j\}$). 
	Similarly to~\Cref{newsuperlemma}, for simplicity of notation we refer to these ``new'' implications using the symbol ``$\Rightarrow_{\text{BF}}$''. 
	Recall that, whenever $x_{ab}\Rightarrow_{\text{BF}} x_{cd}$, we have that $\lambda(a,b)=\lambda(c,d)$ by \BooleanForcing.  
	Note that $x_{uv} \overset{\ast }{\Rightarrow}_{\phi_2'} x_{vu}$ contains at least one ``new'' implication, 
	as otherwise $x_{uv} \overset{\ast }{\Rightarrow}_{\phi_{\text{2SAT}}^{(0)}} x_{vu}$ which is a contradiction by~\Cref{phi-2-0-satisfiable}.

	\medskip
	
	\noindent\textbf{Case I: $x_{uv} \overset{\ast }{\Rightarrow}_{\phi'_2}
		x_{vu}$ contains at least one ``old'' implication.} 
	Consider an ``old'' implication $x_{pq} {\Rightarrow}_{\phi_{\text{2SAT}}^{(0)}} x_{ab}$ within the implication chain $x_{uv} \overset{\ast }{\Rightarrow}_{\phi_2'} x_{vu}$, 
	which is followed by a ``new'' implication (if there is no such pair of consecutive implications, then
	there is one in the contraposition of the implication chain).  By
	\Cref{newsuperlemma}, the ``new'' implication was added by \BooleanForcing\ in Line~\ref{line:add2-1} or
	Line~\ref{line:add2-2}.
	We can assume
	without loss of generality that the ``new'' implication is $x_{ab} {\Rightarrow}_{\text{BF}}
	x_{cb}$ and that $x_{ca}=1$ for some synchronous triangle on the vertices $a,b,c$ (this is the case of Line~\ref{line:add2-1}, Line~\ref{line:add2-2} works analogously). 
	That is, we have
	$\nae(x_{ab},x_{bc},x_{ca})\in\phi_{\text{3NAE}}^{(0)}$. 
	Summarizing, we have that $x_{pq} {\Rightarrow}_{\phi_{\text{2SAT}}^{(0)}}
	x_{ab} {\Rightarrow}_{\text{BF}} x_{cb}$ is contained in $x_{uv} \overset{\ast }{\Rightarrow}_{\phi_2'} x_{vu}$. 
	Furthermore, by construction of
	$\phi_{\text{2SAT}}^{(0)}$, we have that $|\{p,q\}\cap\{a,b,c\}|\le 1$, hence
	we can apply \Cref{lemma-0-after-forcing} and obtain one of the following four scenarios:
	
	\begin{enumerate}
		\item $x_{pq} \overset{\ast }{\Rightarrow}_{\phi_{\text{2SAT}}^{(0)}} x_{cb}$:
		
		In this case we can replace $x_{pq}
		{\Rightarrow}_{\phi_{\text{2SAT}}^{(0)}} x_{ab} {\Rightarrow}_{\text{BF}}
		x_{cb}$ with $x_{pq} {\Rightarrow}_{\phi_{\text{2SAT}}^{(0)}} x_{cb}$ in the
		implication chain $x_{uv} \overset{\ast }{\Rightarrow}_{\phi_{2}'}
		x_{vu}$ to obtain an implication chain from $x_{uv}$ to $x_{vu}$ in $\phi_{2}'$ with strictly
		fewer ``new'' implications, a contradiction.
		
		\item $x_{pq} \overset{\ast }{\Rightarrow}_{\phi_{\text{2SAT}}^{(0)}} x_{bc}$:
		
		Now we have that $x_{pq} {\Rightarrow}_{\phi_{\text{2SAT}}^{(0)}}
		x_{ab}$ and $x_{pq} \overset{\ast }{\Rightarrow}_{\phi_{\text{2SAT}}^{(0)}} x_{bc}$. Then by
		definition of $\phi_{\text{2SAT}}^{(0)}$ we also have that $x_{pq}
		{\Rightarrow}_{\phi_{\text{2SAT}}^{(0)}} x_{ac}$. Recall that we have set $x_{ca}=1$ (which triggered the addition of the implication $x_{ab} {\Rightarrow}_{\text{BF}} x_{cb}$), that is, $x_{ac}=0$. Therefore, by Lines~\ref{line:clean1} and~\ref{line:clean2} of \BooleanForcing,
		we have already set $x_{qp}=1$, i.e.~$x_{pq}=0$, and therefore the implication $x_{pq} {\Rightarrow}_{\phi_{\text{2SAT}}^{(0)}}
		x_{ab}$ does not exist in $\phi_2'$ anymore, which is a contradiction.
		\item $x_{pq} \overset{\ast }{\Rightarrow}_{\phi_{\text{2SAT}}^{(0)}} x_{ca}$:
		
		Now we have that $x_{pq} {\Rightarrow}_{\phi_{\text{2SAT}}^{(0)}}
		x_{ab}$ and $x_{pq} \overset{\ast }{\Rightarrow}_{\phi_{\text{2SAT}}^{(0)}} x_{ca}$. Then by
		definition of $\phi_{\text{2SAT}}^{(0)}$ we also have that $x_{pq}
		{\Rightarrow}_{\phi_{\text{2SAT}}^{(0)}} x_{cb}$. From here it is the same as Case 1.
		\item $x_{pq} \overset{\ast }{\Rightarrow}_{\phi_{\text{2SAT}}^{(0)}} x_{ac}$:   
		Same as Case 2.
	\end{enumerate}

	\medskip

	\noindent\textbf{Case II: $x_{uv} \overset{\ast }{\Rightarrow}_{\phi'_2} x_{vu}$ contains only ``new'' implications.} 
	Similarly to Case II of the proof of~\Cref{newsuperlemma}, we use the notion of \emph{alternating} and \emph{non-alternating} sequences of ``new'' implications. 
	In a nutshell, 
	whenever the sequence $x_{uv} \overset{\ast }{\Rightarrow}_{\text{BF}} x_{vu}$ 
	contains at least one pair of consecutive direct implications of the form $x_{ab} {\Rightarrow}_{\text{BF}} x_{ac} {\Rightarrow}_{\text{BF}} x_{ad}$, 
	or of the form $x_{ba} {\Rightarrow}_{\text{BF}} x_{ca} {\Rightarrow}_{\text{BF}} x_{da}$, 
	the sequence of implications $x_{uv} \overset{\ast }{\Rightarrow}_{\text{BF}} x_{vu}$ is called \emph{non-alternating}; otherwise it is called \emph{alternating}. 
	That is, if $x_{uv} \overset{\ast }{\Rightarrow}_{\text{BF}} x_{vu}$ is \emph{alternating}, then it either has the form
	\begin{equation}
		x_{uv} = x_{u_{1}v_{1}} {\Rightarrow}_{\text{BF}} 
		x_{u_{2}v_{1}} {\Rightarrow}_{\text{BF}} 
		x_{u_{2}v_{2}} \overset{\ast }{\Rightarrow}_{\text{BF}} 
		x_{vu} = x_{v_{1}u_{1}},
		\label{alternating-sequence-UV-1}
	\end{equation}
	or it has the form
	\begin{equation}
		x_{uv} = x_{u_{1}v_{1}} {\Rightarrow}_{\text{BF}} 
		x_{u_{1}v_{2}} {\Rightarrow}_{\text{BF}} 
		x_{u_{2}v_{2}} \overset{\ast }{\Rightarrow}_{\text{BF}} 
		x_{vu} = x_{v_{1}u_{1}}.
		\label{alternating-sequence-UV-2}
	\end{equation}

	We now distinguish the cases where $x_{uv} \overset{\ast }{\Rightarrow}_{\text{BF}} x_{vw}$ is an alternating or a non-alternating sequence of implications. 
	Note that, as all these are ``new'' implications, 
	all edges which are involved in $x_{uv} \overset{\ast }{\Rightarrow}_{\text{BF}} x_{vu}$ have the same label $t$. 
	That is, for every variable $x_{ab}$ that appears in the sequence $x_{uv} \overset{\ast }{\Rightarrow}_{\text{BF}} x_{vu}$ of implications, we have that $\lambda(a,b)=t$.

	\medskip
	
	\noindent\textbf{Case II-A: $x_{uv} \overset{\ast }{\Rightarrow}_{\text{BF}} x_{vu}$ is a non-alternating sequence of implications.} 
	This case can be treated in exactly the same way as Case II-A in the proof of~\Cref{newsuperlemma}, where we just replace ``$x_{vw}$'' with ``$x_{vu}$''. The main idea of the proof is that, if $x_{uv} \overset{\ast }{\Rightarrow}_{\text{BF}} x_{vu}$ is non-alternating, then there exists an implication sequence that contains fewer ``new'' implications, which is a contradiction.
	
	\medskip
	
	\noindent\textbf{Case II-B: $x_{uv} \overset{\ast }{\Rightarrow}_{\text{BF}} x_{vu}$ is an alternating sequence of implications.} 
	First let this sequence be of the form of (\ref{alternating-sequence-UV-1}). 
	As the implication $x_{u_{1}v_{1}} {\Rightarrow}_{\text{BF}} 
	x_{u_{2}v_{1}}$ of (\ref{alternating-sequence-UV-1}) has been added to $\phi_2'$ by \BooleanForcing, it follows that 
	$x_{u_{2}u_{1}}=1$ and $\lambda(u_{1},u_{2})=t$. 
	That is, there is a synchronous triangle on the vertices $\{u_{1},v_{1},u_{2}\}$, and we have the implication sequence 
	$x_{u_{2}v_{1}} \overset{\ast }{\Rightarrow}_{\text{BF}} x_{v_{1}u_{1}}$. Therefore, Lines~\ref{line:add1} and~\ref{line:rmnae1} of \BooleanForcing{} are executed during some iteration $j\geq 1$ (i.e.~in the main phase) of \Cref{temporal-orientation-alg}, which is a contradiction by~\Cref{newsuperlemma}.

	Now let the sequence $x_{uv} \overset{\ast }{\Rightarrow}_{\text{BF}} x_{vu}$ be of the form of (\ref{alternating-sequence-UV-2}). 
	Similarly to Case II-A in the proof of~\Cref{newsuperlemma}, by iteratively applying \Cref{triangle-lemma} (the temporal triangle lemma), we may assume without loss of generality that the first two implications of (\ref{alternating-sequence-UV-2}) are added to $\phi_2'$ by the synchronous triangles on the vertices 
	$\{u_1,v_1,v_2\}$ and 
	$\{u_1,v_2,u_2\}$. 
	Furthermore, as the implications $x_{u_{1}v_{1}} {\Rightarrow}_{\text{BF}} 
	x_{u_{1}v_{2}}$ 
	and 
	$x_{u_{1}v_{2}} {\Rightarrow}_{\text{BF}} 
	x_{u_{2}v_{2}}$
	of (\ref{alternating-sequence-UV-2}) have been added to $\phi_2'$ by \BooleanForcing, it follows that 
	$x_{u_{2}u_{1}}=1$ and $x_{v_{1}v_{2}}=1$. 
	
	Assume that $\{u_{2},v_{1}\} \notin E$ or $\lambda(u_{2},v_{1}) < t$. Then $\phi_{\text{2SAT}}^{(0)}$ by definition contains
	$x_{u_{2}u_{1}} {\Rightarrow}_{\phi_{\text{2SAT}}^{(0)}} x_{v_{1}u_{1}}$. 
	Thus, since $x_{u_{2}u_{1}}=1$, it follows \BooleanForcing\ sets $x_{v_{1}u_{1}}=1$, which is a contradiction to the assumption that the implication $x_{u_{1}v_{1}} {\Rightarrow}_{\text{BF}} x_{u_{1}v_{2}}$ belongs to~$\phi_{2}'$. 
	
	Assume that $\lambda(u_{2},v_{1}) > t$. 
	Then $\phi_{\text{2SAT}}^{(0)}$ by definition contains
	$x_{u_{1}v_{1}} {\Rightarrow}_{\phi_{\text{2SAT}}^{(0)}} x_{u_{2}v_{1}}$ and 
	$x_{u_{2}v_{1}} {\Rightarrow}_{\phi_{\text{2SAT}}^{(0)}} x_{u_{2}v_{2}}$, while both these implications are ``old'' (as these are implications that involve non-synchronous edges). 
	Therefore there exists the sequence of implications 
	$x_{uv} = x_{u_{1}v_{1}} {\Rightarrow}_{\phi_\text{2SAT}^{(0)}} 
	x_{u_{2}v_{1}} {\Rightarrow}_{\phi_\text{2SAT}^{(0)}}
	x_{u_{2}v_{2}} \overset{\ast }{\Rightarrow}_{\text{BF}} 
	x_{vu} = x_{v_{1}u_{1}}$, which contains fewer ``new'' implications, a contradiction.
	
	Finally assume that $\lambda(u_{2},v_{1}) = t$. 
	Then, since $x_{u_{2}u_{1}}=1$ and $x_{v_{1}v_{2}}=1$ and the triangles on the vertices $\{u_1, v_1, u_2\}$ and $\{v_1, u_2, v_2\}$ are synchronous, 
	it follows that we also have the implications $x_{u_{1}v_{1}} {\Rightarrow}_{\text{BF}} x_{u_{2}v_{1}} {\Rightarrow}_{\text{BF}} x_{u_{2}v_{2}}$. 
	Therefore, additionally to (\ref{alternating-sequence-UV-2}), also (\ref{alternating-sequence-UV-1}) is a sequence of (equally many) ``new'' implications from $x_{uv}$ to $x_{vu}$, and thus a contradiction is implied as explained above.
	This completes the proof.
\end{proof}

In the next lemma we prove a strong structural property of our algorithm. Given this property, we will be able to show that, 
if the algorithm does not return ``NO'' during the initialization phase, then the instance is actually a \emph{yes}-instance. That is, 
during all the subsequent iterations $j\geq 1$, the algorithm only constructs a valid transitive orientation, while the decision variant of the problem 
can simply be answered at the end of the initialization phase.

\begin{lemma}\label{lemma:never-say-no}
	For every iteration $j\geq 1$ of~\Cref{temporal-orientation-alg}, every call of \BooleanForcing\ (in Line~\ref{force-1-line} of~\Cref{temporal-orientation-alg}) does \emph{not return ``NO''}.
\end{lemma}

\begin{proof}
	\BooleanForcing\ can possibly return ``NO'' either in Lines~\ref{forcing-phi2-2-line}-\ref{line:2satforce} 
	or in Line~\ref{line-all-equal-no}. 
	First note that, for every call of \BooleanForcing\ in~\Cref{temporal-orientation-alg}, there is a variable $x_{ab}$ which is set to 1 (in Line~\ref{force-1-line} of~\Cref{temporal-orientation-alg}).

	Assume that \BooleanForcing\ returns ``NO'' in Lines~\ref{forcing-phi2-2-line}-\ref{line:2satforce}. 
	Let $(x_{uv} \vee x_{pq})$ be the clause of $\phi'_2$ which is considered in Line~\ref{forcing-phi2-2-line} of \BooleanForcing.
	As all forcings during the execution of \BooleanForcing\ are made by assuming that a specific variable $x_{ab}=1$, we have the following: 
	\begin{itemize}
		\item $x_{ab} \overset{\ast }{\Rightarrow}_{\phi_{2}'} x_{vu}$ \ (as $x_{uv}=0$ in Line~\ref{forcing-phi2-2-line} of \BooleanForcing)
		\item $x_{ab} \overset{\ast }{\Rightarrow}_{\phi_{2}'} x_{qp}$ \ (as $x_{pq}=0$ in Line~\ref{forcing-phi2-2-b-line} of \BooleanForcing)
		\item $x_{vu} {\Rightarrow}_{\phi_{2}'} x_{pq}$ \ (due to the existence of the clause $(x_{uv} \vee x_{pq})$ in $\phi'_2$)
	\end{itemize}
	From the above implications we have that 
	\[
	x_{ab} \overset{\ast }{\Rightarrow}_{\phi_{2}'} x_{vu}
	{\Rightarrow}_{\phi_{2}'} x_{pq}
	\overset{\ast }{\Rightarrow}_{\phi_{2}'} x_{ba},
	\]
	which is a contradiction by~\Cref{super2satlemma}.

	Assume that \BooleanForcing\ returns ``NO'' in Line~\ref{line-all-equal-no}. 
	Then, there exists a clause \textsc{NAE}$(x_{uv}, x_{vw}, x_{wu})$ in $\phi_{\text{3NAE}}^{(0)}$ such that, during the execution of iteration $j$ of~\Cref{temporal-orientation-alg}, we are forced to set each of the variables $x_{uv}, x_{vw}, x_{wu}$ to the same truth value, say without loss of generality, $x_{uv}= x_{vw}= x_{wu}=1$. 
	Furthermore assume without loss of generality that, among these three variables, $x_{uv}$ is the first one that is set to 1 by \BooleanForcing.
	
	Let $x_{uv}$ be set to $1$ at an earlier iteration of \BooleanForcing\ than $x_{vw}$ and $x_{wu}$. 
	Then \BooleanForcing\ adds (in Line~\ref{line:add2-1}) to $\phi_{2}'$ the clause $(x_{wv} \vee x_{uw})$. 
	In this case, when $x_{vw}$ (resp.~$x_{wu}$) is set to $1$ at a subsequent iteration of \BooleanForcing, $x_{uw}$ (resp.~$x_{wv}$) is also set to $1$ (in Lines~\ref{forcing-phi2-2-line}-\ref{line:clean1} of \BooleanForcing). This is a contradiction to our assumption that \BooleanForcing\ sets $x_{uv}= x_{vw}= x_{wu}=1$.

	Let $x_{uv}$ be set to $1$ at the same iteration of \BooleanForcing\ as one of the variables $x_{vw}$ or $x_{wu}$; say, without loss of generality, with $x_{vw}$. 
	Then, as $x_{uv}=x_{vw}=1$, \BooleanForcing\ sets $x_{wu}=0$ (in Line~\ref{line:add2-3}). This is again a contradiction to our assumption that \BooleanForcing\ sets $x_{uv}= x_{vw}= x_{wu}=1$.
\end{proof}

We are now ready to combine all the above technical results to obtain the main result of this section in the next theorem, regarding the correctness and the running time of~\Cref{temporal-orientation-alg}.

\begin{theorem}\label{thm-correctness-runningtime}
	\Cref{temporal-orientation-alg} correctly solves \TTOs\ in polynomial time.
\end{theorem}
\begin{proof}
	First assume that~\Cref{temporal-orientation-alg} returns ``NO''. Due to~\Cref{lemma:never-say-no}, this can only happen in Line~\ref{initialize-return-no-line} of~\Cref{temporal-orientation-alg}, which means that 
	\InitalForcing\ has found a contradiction in $\phi_{\text{3NAE}}^{(0)}\wedge\phi_{\text{2SAT}}^{(0)}$. 
	Thus, clearly $\phi_{\text{3NAE}}^{(0)}\wedge\phi_{\text{2SAT}}^{(0)}$ is not satisfiable, i.e.~$(G,\lambda)$ is not transitively orientable.

	Now assume that~\Cref{temporal-orientation-alg} does not return ``NO''. Than, during its main phase, \Cref{temporal-orientation-alg} repeatedly calls \BooleanForcing, and it repeatedly removes clauses from $\phi_{\text{3NAE}}^{(0)}$, until they are all removed. 
	By Observation~\ref{obs:naesat}, whenever such a clause is removed during the iteration $j\geq 1$ of~\Cref{temporal-orientation-alg}, this clause is satisfied by all satisfying
	assignments of $\phi_{\text{2SAT}}^{(j)}$, 
	and thus it remains satisfied by the truth assignment that is eventually computed by~\Cref{temporal-orientation-alg}. 
	Let $j_0\geq 1$ be the iteration of~\Cref{temporal-orientation-alg}, 
	after which all clauses of $\phi_{\text{3NAE}}^{(0)}$ have been removed. 
	Then $\phi_{\text{2SAT}}^{(j_0)}$ is satisfiable by~\Cref{super2satlemma}, and the subsequent calls of \BooleanForcing\ (in Line~\ref{force-1-line} 
	of~\Cref{temporal-orientation-alg}) provide a satisfying assignment of $\phi_{\text{2SAT}}^{(j_0)}$. 
	
	Let $j_1\geq j_0$ be the last iteration of~\Cref{temporal-orientation-alg}; 
	note that $\phi_{\text{3NAE}}^{(j_1)} \land \phi_{\text{2SAT}}^{(j_1)}$ is empty. 
	Then, in Line~\ref{line-star-3-main-alg}, the algorithm gives the arbitrary truth value $x_i=1$ to every variable $x_i$ 
	which did not yet get any truth value yet. This is a correct decision as all these variables are not involved in any Boolean constraint 
	of $\phi_{\text{3NAE}}^{(j_1)} \land \phi_{\text{2SAT}}^{(j_1)}$ (which is empty). Finally, the algorithm orients in Line~\ref{line-star-4-main-alg}
	all edges of $G$ according to the corresponding truth assignment. The returned orientation $F$ of $(G,\lambda)$ is temporally transitive 
	as every variable was assigned a truth value according to the Boolean constraints throughout the execution of the algorithm.

	Summarizing, the truth assignment produced by~\Cref{temporal-orientation-alg} satisfies $\phi_{\text{3NAE}}^{(0)}\wedge\phi_{\text{2SAT}}^{(0)}$, and thus the algorithm returns a valid temporally transitive orientation of the input temporal graph $(G,\lambda)$. This completes the proof of correctness of~\Cref{temporal-orientation-alg}.
	
	Lastly, we prove that \Cref{temporal-orientation-alg} runs in polynomial time. 
	The $\Lambda$-implication classes of $(G,\lambda)$ can be clearly computed by \Cref{edge-variables-alg} in polynomial time. 
	\BooleanForcing\ iteratively adds and removes clauses from the 2SAT formula $\phi'_2$, while 
	it can only remove clauses from the 3NAE formula $\phi'_3$. 
	Whenever a clause is added to $\phi'_2$, a clause of $\phi'_3$ is removed. Therefore, as the initial 3NAE formula $\phi_3$ has at most polynomially-many 
	clauses, we can add clauses to $\phi'_2$ only polynomially-many times. 
	In all remaining steps,~\BooleanForcing\ just 
	checks clauses of $\phi'_2$ and $\phi'_3$ and it forces certain truth values to variables, and thus the total running time of~\BooleanForcing\ 
	is polynomial.
	Furthermore, in \InitalForcing\ and \Cref{temporal-orientation-alg} (the main algorithm) 
	the \BooleanForcing\ subroutine is only invoked at most four times for every variable in $\phi_{\text{3NAE}}^{(0)} \wedge \phi_{\text{2SAT}}^{(0)}$. Hence, we have an overall polynomial running time.
\end{proof}

\section{Temporal Transitive Completion}\label{completion-sec}

We now study the computational complexity of \TTC. 
In the static case, the so-called \emph{minimum comparability completion} problem, 
i.e.~adding the smallest number of edges to a static graph to turn it into a comparability graph, 
is known to be NP-hard~\cite{hakimi1997orienting}. 
Note that minimum comparability completion on static graphs is a special case of \TTCs\, and thus it follows that \TTCs\ is NP-hard too. 
Our other variants, however, do not generalize static comparability
completion in such a straightforward way. Note that for \textsc{Strict} \TTCs\
we have that the corresponding recognition problem \StrictTTOs\ is NP-complete
(\Cref{thm:sgehard}), hence it follows directly that \textsc{Strict} \TTCs\ is NP-hard. 
For the remaining two variants of our problem, we show in the following
that they are also NP-hard, giving the result that all four variants of \TTCs\ are NP-hard. 
Furthermore, we present a polynomial-time algorithm for all four problem variants for
the case that all edges of underlying graph are oriented, see~\Cref{thm:full-oriented}. This allows directly
to derive an FPT algorithm for the number of unoriented edges as a parameter.

\begin{theorem}
	All four variants of \TTCs\ are NP-hard, even when the input temporal graph is completely unoriented.
\end{theorem}
\begin{proof}
	\newcommand{\vx}[1]{\tn{v}_{#1}}
	\newcommand{\wx}[2]{\tn{w}_{#1, #2}}

	\begin{figure}
		\centering
		\begin{tikzpicture}[scale=1]
			\newcommand{\vargadget}[1]{
				\draw[every node/.style={vertex}]
				(-1,0)  node[label=below:$\vx{#1}$] ({#1}) {}
				(1,0)   node[label=below:$\vx{\overline{#1}}$] (not{#1})  {}
				(-1,2)  node (not{#1}1) {}
				(1,2)   node ({#1}1)  {}
				(-0.7,1.4) node (not{#1}2) {}
				(0.7,1.4)  node ({#1}2) {}
				;
				
				\draw[diredge,every node/.style={timelabel}]
				({#1}) edge node {$1$} (not{#1})
				({#1}1) edge node {$4$} (not{#1})
				({#1}2) edge (not{#1})
				({#1}) edge node {$4$} (not{#1}1)
				({#1}) edge (not{#1}2)
				({#1}1) edge node {$1$} (not{#1}1)
				({#1}2) edge (not{#1}2)
				;
			}
			
			\newcommand{\vargadgetleft}[1]{
				\draw[every node/.style={vertex}]
				(-1,0)  node[label=below:$\vx{#1}$] ({#1}) {}
				(1,0)   node[label=left:$\vx{\overline{#1}}$] (not{#1})  {}
				(-1,2)  node (not{#1}1) {}
				(1,2)   node ({#1}1)  {}
				(-0.7,1.4) node (not{#1}2) {}
				(0.7,1.4)  node ({#1}2) {}
				;
				
				\draw[diredge,every node/.style={timelabel}]
				({#1}) edge node {$1$} (not{#1})
				({#1}1) edge node {$4$} (not{#1})
				({#1}2) edge (not{#1})
				({#1}) edge node {$4$} (not{#1}1)
				({#1}) edge (not{#1}2)
				({#1}1) edge node {$1$} (not{#1}1)
				({#1}2) edge (not{#1}2)
				;
			}
			\newcommand{\vargadgetright}[1]{
				\draw[every node/.style={vertex}]
				(-1,0)  node[label=right:$\vx{#1}$] ({#1}) {}
				(1,0)   node[label=below:$\vx{\overline{#1}}$] (not{#1})  {}
				(-1,2)  node (not{#1}1) {}
				(1,2)   node ({#1}1)  {}
				(-0.7,1.4) node (not{#1}2) {}
				(0.7,1.4)  node ({#1}2) {}
				;
				
				\draw[diredge,every node/.style={timelabel}]
				({#1}) edge node {$1$} (not{#1})
				({#1}1) edge node {$4$} (not{#1})
				({#1}2) edge (not{#1})
				({#1}) edge node {$4$} (not{#1}1)
				({#1}) edge (not{#1}2)
				({#1}1) edge node {$1$} (not{#1}1)
				({#1}2) edge (not{#1}2)
				;
			}
			\begin{scope}[shift={(90:2)}]
				\vargadget{x}
			\end{scope}
			\begin{scope}[shift={(210:2)},rotate=120]
				\tikzset{diredge/.append style={<-}} 
				\vargadgetleft{y}
			\end{scope}
			\begin{scope}[shift={(330:2)},rotate=-120]
				\vargadgetright{z}
			\end{scope}
			
			\node[vertex,label=left:$\wx{x}{\overline{y}}$] (w1) at (150:3) {};
			\draw[diredge] (x) --node[timelabel] {$2$} (w1);
			\draw[diredge] (not{y}) --node[timelabel] {$3$} (w1);
			
			\node[vertex,label=right:$\wx{\overline{x}}{z}$] (w2) at (30:3) {};
			\draw[diredge] (w2) --node[timelabel] {$2$} (not{x});
			\draw[diredge] (z) --node[timelabel] {$3$} (w2);
			
			\node[vertex,label=above:$\wx{\overline{y}}{\overline{z}}$] (w3) at (0,0) {};
			\draw[diredge] (not{y}) --node[timelabel] {$2$} (w3);
			\draw[diredge] (w3) --node[timelabel] {$3$} (not{z});
			
			\draw[diredge,specialedge] (not{y}) --node[timelabel] {$5$} (not{z});
			
		\end{tikzpicture}
		\caption{Temporal graph constructed from the formula $(x \Rightarrow \overline{y}) \land (\overline{x} \Rightarrow z) \land (\overline{y} \Rightarrow \overline{z})$
			for $k=1$ with orientation corresponding to the assignment $x = \true$, $y = \false$, $z = \true$.
			Since this assignment does not satisfy the third clause, the dashed blue edge is required to make the graph temporally transitive.}
		\label{fig:completion-hardness}
	\end{figure}
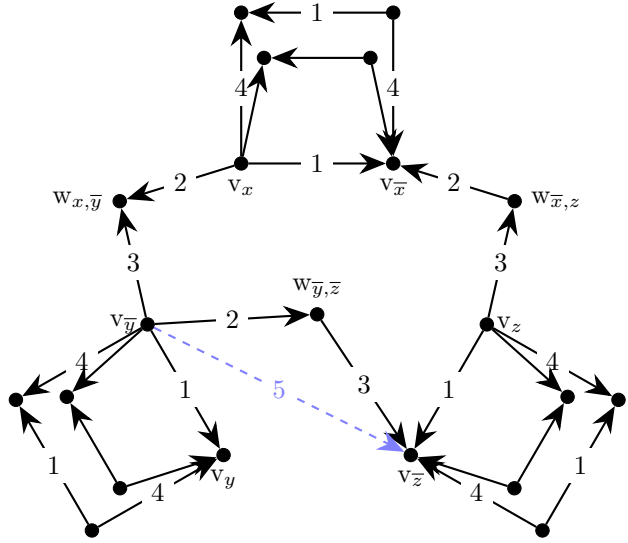
	
	We give a reduction from the NP-hard \maxtwosat{} problem \cite{GAREY1976237}. 
	\problemdef{\maxtwosat}
	{ A boolean formula $\phi$ in implicative normal form\footnotemark{}  and an integer $k$.}
	{Is there an assignment of the variables which satisfies at least $k$ clauses in $\phi$?}
	\footnotetext{i.e.~a conjunction of clauses of the form $(a \Rightarrow b)$ where $a, b$ are literals.}
	We only describe the reduction from \maxtwosat{} to \TTCs. 
	However, the same construction can be used to show NP-hardness of the other variants. 
	
	Let $(\phi,k)$ be an instance of \maxtwosat\ with $m$ clauses.
	We construct a temporal graph $\TG$ as follows.
	For each variable $x$ of $\phi$ 
	we add two vertices denoted $\vx{x}$ and 
	$\vx{\overline{x}}$, connected by an edge with label~$1$.
	Furthermore, for each $1 \leq i \leq m-k+1$ we add two vertices $\vx{x}^i$ and $\vx{\overline{x}}^i$ connected by an edge with label~$1$.  
	We then connect $\vx{x}^i$ with $\vx{\overline{x}}$ and $\vx{\overline{x}}^i$ with $\vx{x}$ using two edges labeled~$4$.
	Thus $\vx{x}, \vx{\overline{x}}, \vx{x}^i, \vx{\overline{x}}^i$ is a 4-cycle whose edges alternating between $1$ and $4$.
	Afterwards, for each clause $(a \Rightarrow b)$ of $\phi$ with $a, b$ being literals,
	we add a new vertex $\wx{a}{b}$.
	Then we connect $\wx{a}{b}$ to $\vx{a}$ by an edge labeled~$2$
	and to $\vx{b}$ by an edge labeled~$3$.
	Consider \Cref{fig:completion-hardness} for an illustration.
	Observe that $\TG$ can be computed in polynomial time.
	
	We claim that $(\TG=(G,\lambda),\emptyset, m-k)$ is a yes-instance of \TTCs\ if and only if $\phi$ has a truth assignment satisfying $k$ clauses.
	
	For the proof, begin by observing that $\TG$ does not contain any triangle.
	Thus an orientation of $\TG$ is (weakly) (strict) transitive if and only if it does not have any oriented temporal 2-path, i.e.~a temporal path of two edges with both edges being directed forward.
	We call a vertex $v$ of $\TG$ \emph{happy} about some orientation if $v$ is not the center vertex of an oriented temporal 2-path.
	Thus an orientation of $\TG$ is transitive if and only if all vertices are happy.

	\medskip
	
	\noindent \textbf{($\Leftarrow$):}
	Let $\alpha$ be a truth assignment to the variables (and thus literals) of $\phi$ satisfying $k$ clauses of $\phi$.
	For each literal $a$ with $\alpha(a) = \true$, orient all edges such that they point away from $\vx{a}$ and $\vx{a}^i$, $1 \leq i \leq m-k+1$.
	For each literal $a$ with $\alpha(a) = \false$, orient all edges such that they point towards $\vx{a}$ and $\vx{a}^i$, $1 \leq i \leq m-k+1$.
	Note that this makes all vertices $\vx{a}$ and $\vx{a}^i$ happy.
	Now observe that a vertex $\wx{a}{b}$ is happy unless its edge with $\vx{a}$ is oriented towards $\wx{a}{b}$ 
	and its edge with $\vx{b}$ is oriented towards $\vx{b}$.
	In other words, $\wx{a}{b}$ is happy if and only if $\alpha$ satisfies the clause $(a \Rightarrow b)$.
	Thus there are at most $m-k$ unhappy vertices.
	For each unhappy vertex $\wx{a}{b}$, we add a new oriented edge from $\vx{a}$ to $\vx{b}$ with label~$5$.
	Note that this does not make $\vx{a}$ or $\vx{b}$ unhappy as all adjacent edges are directed away from $\vx{a}$ and towards $\vx{b}$.
	The resulting temporal graph is transitively oriented.

	\medskip
	
	\noindent \textbf{($\Rightarrow$):} 
	Now let a transitive orientation $F'$ 
	of $\TG'=(G',\lambda')$ be given, where $\TG'$ is obtained from $\TG$ by adding at most $m-k$ time edges.
	Clearly we may also interpret $F'$ 
	as an orientation induced of $\TG$.
	Set $\alpha(x)= \true{}$ if and only if the edge between $\vx{x}$ and $\vx{\overline{x}}$ is oriented towards $\vx{\overline{x}}$.
	We claim that this assignment $\alpha$ satisfies at least $k$~clauses of $\phi$.
	
	First observe that for each variable~$x$ and $1 \leq i \leq m-k+1$, $F'$ 
	is a transitive orientation of the 4-cycle
	$\vx{x}, \vx{\overline{x}}, \vx{x}^i, \vx{\overline{x}}^i$ if and only if the edges are oriented alternatingly.
	Thus, for each variable, at least one of these $k+1$ 4-cycles is oriented alternatingly.
	In particular, for every literal $a$ with $\alpha(a)= \true{}$, there is an edge with label~$4$ that is oriented away from $\vx{a}$.
	Conversely, if $\alpha(b) = \false{}$, then there is an edge with label~$1$ oriented towards $\vx{b}$ 
	(this is simply the edge from $\vx{\overline{b}}$).
	
	This implies that every edge with label $2$ or $3$ oriented from some vertex $\wx{c}{d}$ (where either $a=c$ or $a=d$) towards $\vx{a}$ with $\alpha(a) = \true{}$ requires $E(G') \setminus E(G)$ 
	to contain an edge from $\wx{c}{d}$ to some $\vx{\overline{a}}^i$.
	Analogously every edge with label $2$ or $3$ oriented from $\vx{a}$ with $\alpha(a) = \false{}$ to some $\wx{c}{d}$ requires $E(G') \setminus E(G)$ to contain an edge from $\vx{\overline{a}}$ to $\wx{c}{d}$.
	
	Now consider the alternative orientation $F''$ obtained from $\alpha$ as detailed in the converse orientation of the proof.
	For each edge between $\vx{a}$ and $\wx{c}{d}$ where $F'$ and $F''$ disagree,
	$F''$ might potentially require $E(G') \setminus E(G)$ to contain the edge $\vx{c}\vx{d}$ (labeled~$5$, say),
	but in turn saves the need for some edge $\wx{c}{d}\vx{\overline{a}}^i$ or $\vx{\overline{a}}\wx{c}{d}$, respectively.
	Thus, overall, $F''$ requires at most as many edge additions as~$F'$, which are at most~$m-k$.
	As we have already seen in the converse direction, 
	$F''$ requires exactly one edge to be added for every clause of $\phi$ which is not satisfied.
	Thus, $\alpha$ satisfies at least $k$~clauses of $\phi$.
\end{proof}

We now show that \TTCs\ can be solved in polynomial time, if all edges are
already oriented, as the next theorem states. 
While we only discuss the algorithm for \TTCs\, the algorithm only needs
marginal changes to work for all other variants.
\begin{theorem}
	\label{thm:full-oriented}
	An instance $(\TG,F,k)$ of \TTCs\, where $\TG=(G,\lambda)$ and $G=(V,E)$, 
	can be solved in $O(m^2)$ time if $F$ is an orientation of $E$,
	where $m=|E|$.
\end{theorem}
The actual proof of \Cref{thm:full-oriented} is deferred  to the end of this section.
The key idea for the proof is based on the following definition. 
Assume a temporal graph $\TG$ and an orientation $F$ of $\TG$ to be given. 
Let $G' = (V, F)$ be the underlying graph of $\TG$ with its edges directed according to $F$. 
We call a (directed) path $P$ in $G'$ \emph{tail-heavy} if the time-label of its last edge 
is largest among all edges of $P$, and we define $t(P)$ to be the time-label of that last edge of $P$. 
For all $u,v \in V$, denote by $T_{u,v}$ the maximum value~$t(P)$ over all tail-heavy $(u,v)$-paths~$P$ of length at least $2$ in $G'$; 
if such a path does not exist then $T_{u,v} = \bot$. 
If the temporal graph $\TG$ with orientation $F$ can be completed to be transitive, then
adding the time edges of the set 
\begin{align*}
	X(\TG,F) \coloneqq \left\{ (uv, T_{u,v}) \mvert T_{u,v} \neq \bot \right\},
\end{align*} 
which are not already present in $\TG$ is an optimal way to do so. Consider \Cref{fig:tail-heavy} for an example.
\begin{figure}
	\centering
	\begin{tikzpicture}[scale=1]
		\draw[every node/.style={vertex}]
		(0,0) node[label=left:$a$] (a) {}
		++(2,0) node[label=$b$] (b) {}
		++(2,0) node[label=$c$] (c) {}
		++(2,0) node[label=right:$d$] (d) {}
		;
		
		\draw[diredge,every node/.style={timelabel}]
		(a) edge node {$2$} (b)
		(b) edge node {$1$} (c)
		(c) edge node {$3$} (d)
		;
		
		\draw[diredge,specialedge,shorten >=5pt,shorten <=5pt]
		(b) edge[bend right] node[timelabel] {$T_{b,d}=3$} (d)
		(a) edge[bend left] node[timelabel] {$T_{a,d}=3$} (d);

	\end{tikzpicture}
	\caption{Example of a tail-heavy path.}
	\label{fig:tail-heavy}
\end{figure}
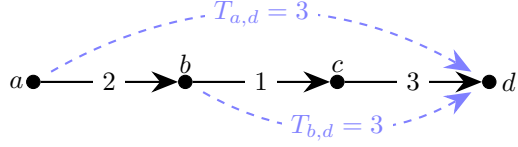

\begin{lemma}
	\label{obs:X-poly-time}
	The set $X(\TG,F)$ can be computed in $O(m^2)$ time, where $\TG$ is a temporal graph
	with $m$~time-edges and $F$ an orientation of $\TG$.
\end{lemma}
\begin{proof}
	For each edge~$vw$,
	we can take $G'$ (defined above),
	remove $w$ and all arcs whose label is larger than $\lambda(v,w)$,
	and do a depth-first-search from $v$ to find all vertices~$u$ which can reach~$v$ in the resulting graph.
	Each of these then has $T_{u,w} \geq \lambda(v, w)$.
	By doing this for every edge~$vw$, we obtain~$T_{u,w}$ for every vertex pair $u,w$. 
	The overall running time is clearly $\bigO(m^2)$.
\end{proof}

Until the end of this section we are only considering the instance $(\TG,F,k)$ of \TTCs, 
where $\TG=(G,\lambda)$, $G=(V,E)$, and $F$ is an orientation of $\TG$.
Hence, we can say a set~$X$ of oriented time-edges is a \emph{solution} to $I$ if 
$X' := \{ \{u,v\} \mid (uv,t) \in X\}$ is disjoint from $E$, satisfies $|X| = |X'| \leq k$,
and $F':= F \cup \{ uv \mid (uv,t) \in X\}$ is a transitive orientation of 
the temporal graph $\TG+X := ((V,E \cup X'),\lambda')$,
where $\lambda'(e) := \lambda(e)$ if $e\in E$ and 
$\lambda'(u,v) := t$ if $X$ contains $(uv,t)$ or $(vu, t)$.

The algorithm we use to show \Cref{thm:full-oriented} will use $X(\TG,F)$ to
construct a solution (if there is any) of a given instance $(\TG,F,k)$ of \TTCs\ where $F$ is a orientation of $E$.
To prove the correctness of this approach, we make use of the following.
\begin{lemma}
	\label{lem:add-this}
	Let $I = (\TG = (G,\lambda),F,k)$ be an instance of \TTCs, where $G=(V,E)$ and $F$ is an orientation of $E$
	and $X$ an solution for $I$.
	Then, for any $(vu,T_{v,u}) \in X(\TG,F)$ there is a $(vu,t)$ in $\TG+X$ with $t \geq T_{v,u}$.
\end{lemma}
\begin{proof}
	Let $(v_0 v_\ell,T_{v_0,v_{\ell}})\in X(\TG,F)$, and $G'=(V,F)$.
	Hence, there is a tail-heavy $(v_0,v_\ell)$-path $P$ in $G'$ of length $\ell\geq 2$.
	If $\ell=2$, then clearly $\TG+X$ must contain the time edge $(v_1 v_\ell,t)$ such that $t \geq T_{v_1,v_\ell}$.
	Now let $\ell>2$ and $V(P):=\{v_i \mid i \in \{0,1,\dots,\ell\}\}$ and $E(P) = \{ v_{i-1} v_i \mid i \in [\ell] \}$.
	Since there is a tail-heavy $(v_{\ell-2},v_\ell)$-path in $G'$ of length~2,
	$\TG+X$ must contain a time-edge $(v_{\ell-2} v_\ell,t)$ with $t \geq T_{v_{0},v_\ell}$.
	Therefore, the (directed) underlying graph of $\TG+X$ contains a tail-heavy $(v_0,v_\ell)$-path of length $\ell-1$.
	By induction, $\TG+X$ must contain the time edge $(v_1 v_\ell,t')$ such that $t' \geq t \geq T_{v_0,v_\ell}$.
\end{proof}
Form \Cref{lem:add-this}, it follows that we can use $X(\TG,F)$ to identify \no-instances in some cases.
\begin{corollary}
	\label{lem:clear-no}
	Let $I = (\TG = (G,\lambda),F,k)$ be an instance of \TTCs, where $G=(V,E)$ and $F$ is an orientation of $E$.
	Then, $I$ is a \no-instance, if for some $v,u \in V$
	\begin{enumerate}
		\item there are time-edges $(vu,t) \in X(\TG,F)$ and $(uv,t') \in X(\TG,F)$, 
		\item there is an edge $uv \in F$ such that $(vu,T_{v,u}) \in X(\TG,F)$, or
		\item there is an edge $vu \in F$ such that $(vu,T_{v,u}) \in X(\TG,F)$ with $\lambda(v,u) < T_{v,u}$.
	\end{enumerate}
\end{corollary}

We are now ready to prove \Cref{thm:full-oriented}.
\begin{proof}[Proof of \Cref{thm:full-oriented}]
	Let $I=(\TG=(G,\lambda),F,k)$ be an instance of \TTCs, where $F$ is a orientation of $E$.
	First we compute $X(\TG,F)$ in polynomial time, see \Cref{obs:X-poly-time}.
	Let $Y=\{ (vu,t) \in X(\TG,F) \mid \{v,u\} \not\in E\}$
	and report that $I$ is a \no-instance 
	if $|Y| > k$ or one of the conditions of \Cref{lem:clear-no} holds true.
	Otherwise report that $I$ is a \yes-instance.
	This gives an overall running time of $O(m^2)$.
	
	Clearly, if one of the conditions of \Cref{lem:clear-no} holds true, then $I$ is a \no-instance.
	Moreover, by \Cref{lem:add-this} any solution contains at least $|Y|$ time edges.
	Thus, if $|Y| > k$, then $I$ is a \no-instance.
	
	If we report that $I$ is a \yes-instance, 
	then we claim that $Y$ is a solution for $I$.
	Let $F' \supseteq F$ be a orientation of $\TG+Y$.
	Assume towards a contradiction that $F'$ is not transitive.
	Then, there is a temporal path $((vu,t_1),(uw,t_2))$ in $\TG+Y$ such that
	there is no time-edge $(uw,t)$ in $\TG+Y$, with $t \geq t_2$.
	By definition of $X(\TG,F)$, the directed graph $G'=(V,F)$ contains 
	a tail-heavy $(v,u)$-path $P_1$ with $t_1 = t(P_1)$ and
	a tail-heavy $(u,w)$-path $P_2$ with $t_2 = t(P_2)\geq t_1$.
	By concatenation of $P_1$ and $P_2$, we obtain that the $G'$ contains a $(v,w)$-path $P'$ of length at least two 
	such that $t_2 = t(P')$.
	Thus, $t_2 \leq T_{v,w}$ and $(vw,T_{v,w}) \in X(\TG)$---a contradiction.
\end{proof}

Using \Cref{thm:full-oriented} we can now prove that \TTCs\ is fixed-parameter tractable (FPT) with respect to the number of 
unoriented edges in the input temporal graph $\TG$. 
\begin{corollary}
	Let $I=(\TG=(G,\lambda),F,k)$ be an instance of \TTCs, where $G=(V,E)$. 
	Then $I$ can be solved in $O(2^q\cdot m^2)$, where $q=|E|-|F|$ and $m$ the number of time edges.
\end{corollary}
\begin{proof}
	Note that there are $2^q$ ways to orient the $q$ unoriented edges. 
	For each of these $2^q$ orientations of these $q$ edges, we obtain a fully oriented temporal graph. 
	Then we can solve \TTCs\ on each of these fully oriented graphs in $O(m^2)$
	time by \Cref{thm:full-oriented}.
	Summarizing, we can solve \TTCs\ on $I$ in $2^q \cdot m^2$ rime.
\end{proof}

\section{Deciding Multilayer Transitive Orientation}\label{multilayer-sec}

In this section we prove that \MTO\ is NP-complete, even if every edge of the given temporal graph has at most two labels. 
Recall that this problem asks for an orientation $F$ of a temporal graph $\TG=(G,\lambda)$ 
(i.e.~with exactly one orientation for each edge of $G$) such that, 
for every ``time-layer'' $t\geq 1$, the (static) oriented graph defined by the edges having time-label $t$ is transitively oriented in $F$.
As we discussed in~\Cref{prelim-sec}, this problem makes more sense when every edge of $G$ 
potentially has multiple time-labels, therefore we assume here that the time-labeling function is $\lambda :E\rightarrow 2^{\mathbb{N}}$.

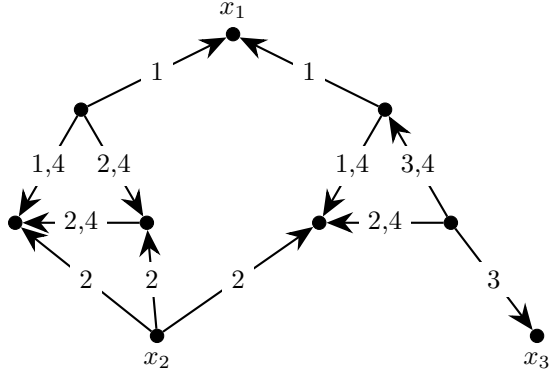
\begin{figure}[t]
	\centering
	\begin{tikzpicture}[scale=1]
		\draw[every node/.style={vertex}]
		(90:1)  node (y11)  {}
		(210:1) node (y12)  {}
		(330:1) node (y13)  {}
		;
		
		\draw[diredge, every node/.style={timelabel}]
		(y11) edge node {1,4} (y12)
		(y11) edge node {2,4} (y13)
		(y13) edge node {2,4} (y12)
		;
		
		\draw[xshift=4cm,every node/.style=vertex]
		(90:1)  node (y21) {}
		(210:1) node (y22) {}
		(330:1) node (y23) {}
		;
		
		\draw[diredge,every node/.style={timelabel}]
		(y21) edge node {1,4} (y22)
		(y23) edge node {3,4} (y21)
		(y23) edge node {2,4} (y22)
		;
		
		\node[vertex,label=$x_1$] (a) at (2,2) {};
		\node[vertex,label=below:$x_2$] (b) at (1,-2) {};
		\node[vertex,label=below:$x_3$] (c) at (6,-2) {};
		
		\draw[diredge, every node/.style={timelabel}]
		(y11) edge node {1} (a)
		(y21) edge node {1} (a)
		(b) edge node {2} (y12)
		(b) edge node {2} (y13)
		(b) edge node {2} (y22)
		(y23) edge node {3} (c)
		;
	\end{tikzpicture}
	\caption{Temporal graph constructed from the formula $\nae(x_1, x_2, x_2) \land \nae(x_1, x_2, x_3)$
		and orientation corresponding to setting $x_1 = \false$, $x_2 = \true$, and $x_3 = \false$.
		Each attachment vertex is at the clockwise end of its edge.}
	\label{fig:multilayer-reduction}
\end{figure}

\begin{theorem}
	\MTOs\ is NP-complete, even on temporal graphs with at most two labels per edge.
\end{theorem}
\begin{proof}
	\newcommand{\vx}[1]{\tn{v}_{#1}}
	\newcommand{\ex}[1]{\tn{a}_{#1}}
	
	We give a reduction from monotone \NAESat{}, which is known to be NP-hard \cite{Schaefer78}.
	So let $\phi = \bigwedge_{i=1}^m \nae(y_{i,1}, y_{i,2}, y_{i,3})$ be a monotone \NAESat{} instance
	and $X := \{x_1, \dots, x_n\} := \bigcup_{i=1}^m \{y_{i,1}, y_{i,2}, y_{i,3}\}$ be the set of variables.
	
	Start with an empty temporal graph $\TG$.
	For every clause $\nae(y_{i,1}, y_{i,2}, y_{i,3})$, add to $\TG$ a triangle on three new vertices
	and label its edges $\ex{i,1}, \ex{i,2}, \ex{i,3}$.
	Give all these edges label~$n+1$.
	For each of these edges, select one of its endpoints to be its \emph{attachment vertex}
	in such a way that no two edges share an attachment vertex.
	Next, for each $1 \leq i \leq n$, add a new vertex $\vx{i}$.
	Let $A_i := \{\ex{i,j} \mid y_{i,j} = x_i\}$.
	Add the label~$i$ to every edge in $A_i$ and connect its attachment vertex to $\vx{i}$ with an edge labeled~$i$.
	See also \Cref{fig:multilayer-reduction}.
	
	We claim that $\TG$ is a \yes-instance of \MTOs\ if and only if $\phi$ is satisfiable.
	
	\medskip\noindent
	\textbf{($\Leftarrow$):}
	Let $\alpha: X \to \{\true, \false\}$ be an assignment satisfying $\omega$.
	For every $x_i \in X$, orient all edges adjacent to $\vx{i}$ away from $\vx{i}$ if $\alpha(x_i) = \true$ and towards $\vx{i}$ otherwise.
	Then, orient every edge~$\ex{i,j}$ towards its attachment vertex if $\alpha(y_{i,j}) = \true$ and away from it otherwise.
	
	Note that in the layers $1$~through~$n$ every vertex either has all adjacent edges oriented towards it or away from it.
	Thus these layers are clearly transitive.
	It remains to consider layer~$n+1$ which consists of a disjoint union of triangles.
	Each such triangle $\ex{i,1}, \ex{i,2}, \ex{i,3}$ is oriented non-transitively (i.e.~cyclically) if and only if 
	$\alpha(y_{i,1}) = \alpha(y_{i,2}) = \alpha(y_{i,3})$, which never happens if $\alpha$ satisfies $\phi$.
	
	\medskip\noindent
	\textbf{($\Rightarrow$):}
	Let $\omega$ be an orientation of the underlying edges of $\TG$ such that every layer is transitive.
	Since they all share the same label~$i$, the edges adjacent to $\vx{i}$ must be all oriented towards or all oriented away from $\vx{i}$.
	We set $\alpha(x_i) = \false$ in the former and $\alpha(x_i) = \true$ in the latter case.
	This in turn forces each edge $\ex{i,j}$ to be oriented towards its attachment vertex if and only if $\alpha(\ex{i,j}) = \true$.
	Therefore, every clause $\nae(y_{i,1}, y_{i,2}, y_{i,3})$ is satisfied, since the three edges $\ex{i,1}, \ex{i,2}, \ex{i,3}$ form a triangle in layer~$n+1$ and can thus not be oriented cyclically (i.e.~all towards or all away from their respective attachment vertices).
\end{proof}

\section{Conclusion}
We introduced and studied four natural variants of temporal graph transitivity. 
Although these four variants look superficially similar, they turn out to have massive differences in their computational complexity.
Two variants (\StrongTTOs\ and \StrongStrictTTOs) are solvable by straightforward reductions to 2SAT. 
For \TTOs\ we provided a technically involved polynomial-time algorithm which solves the problem by first reducing it to 
the satisfiability of a mixed Boolean formula (having both clauses with three and with two literals) and by then using 
a series of structural properties to devise a polynomial-time algorithm. 
That is, we reduce \TTOs\ to the satisfiability problem of a special subclass of mixed Boolean formulas which turns out to be efficiently solvable. 
We leave it open for future research whether a compact set of conditions can be given which define this subclass of mixed Boolean formulas, as this might be of independent interest.  
The last variant \StrictTTOs\ turns out to be NP-hard.

We further studied the ``completion''-problem corresponding to each of the four
temporal transitivity variants, that is, finding the minimum number of time
edges that need to be added to a given temporal graph to make it transitive. We
show for all four completion problem variants that they are NP-hard. However if
the edges of the temporal input graph are already oriented, we obtain
polynomial-time solvability which we can easily generalize to an FPT-algorithm
for the number of unoriented edges as a parameter. Here, we in particular leave
the parameterized complexity with respect to the solution size or other parameters open for future research. 
Lastly, we investigate a natural extension of transitivity to multilayer
graphs and show that deciding whether a given multilayer graph is transitive is NP-hard.

\end{document}